%% file: main.tex
\documentclass[USenglish,oneside,twocolumn]{article}

\usepackage[utf8]{inputenc}%
\usepackage[big]{dgruyter_NEW}
 
\DOI{foobar}

\cclogo{\includegraphics{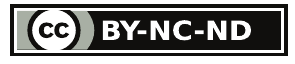}}

\def\Comments{0} %
\input{macros}
\input{local-macros}

\begin{document}

\author[1]{Vadym Doroshenko}

\author[2]{Badih Ghazi}

\author[3]{Pritish Kamath}

\author[4]{Ravi Kumar}

\author[5]{Pasin Manurangsi}

\affil[1]{Google, E-mail: dvadym@google.com}

\affil[2]{Google Research, E-mail: badihghazi@gmail.com}

\affil[3]{Google Research, E-mail: pritish@alum.mit.edu}

\affil[4]{Google Research, E-mail: ravi.k53@gmail.com}

\affil[5]{Google Research, E-mail: pasin@google.com}

\title{\huge Connect the Dots: Tighter Discrete Approximations of Privacy Loss Distributions}

\runningtitle{Connect the Dots: Tighter Discrete Approximations of Privacy Loss Distributions}

\begin{abstract}
{The privacy loss distribution (PLD) provides a tight characterization of the privacy loss of a mechanism in the context of differential privacy (DP). Recent work~\cite{meiser2018tight,koskela2020computing,koskela2021tight,koskela2021computing} has shown that PLD-based accounting allows for tighter $(\eps, \delta)$-DP guarantees for many popular mechanisms compared to other known methods. A key question in PLD-based accounting is how to approximate any (potentially continuous) PLD with a PLD over any specified discrete support. 
\\
We present a novel approach to this problem. Our approach supports both {\em pessimistic} estimation, which overestimates the hockey-stick divergence (i.e., $\delta$) for any value of $\eps$, and {\em optimistic} estimation, which underestimates the hockey-stick divergence. Moreover, we show that our pessimistic estimate is the \emph{best} possible among all pessimistic estimates.
Experimental evaluation shows that our approach can work with much larger discretization intervals while keeping a similar error bound compared to previous approaches and yet give a better approximation than an existing method~\cite{meiser2018tight}.}
\end{abstract}
\keywords{privacy loss distribution, pessimistic approximation, optimistic approximation, privacy accounting, composition}

\journalname{Proceedings on Privacy Enhancing Technologies}
\DOI{10.2478/popets-2022-0095}
\startpage{1}
\received{2022-02-28}
\revised{2022-06-15}
\accepted{2022-06-16}

\journalyear{2022}
\journalvolume{}
\journalissue{4}

\maketitle

\section{Introduction}

Differential privacy (DP) \cite{dwork06calibrating,dwork2006our} has become widely adopted as a notion of privacy in analytics and machine learning applications, leading to numerous practical deployments including in industry \cite{erlingsson2014rappor,CNET2014Google, greenberg2016apple,dp2017learning, ding2017collecting} and government agencies \cite{abowd2018us}. The DP guarantee of a (randomized) algorithm is parameterized by two real numbers $\eps > 0$ and $\delta \in [0,1]$; the smaller these values, the more private the algorithm. 

 The appeal of DP stems from the strong privacy that it guarantees (which holds even if the adversary controls the inputs of all other users in the database), and from its nice mathematical properties. These include \emph{composition}, whose \emph{basic} form \cite{dwork2006our} says that executing  an $(\eps_1, \delta_1)$-DP algorithm and an $(\eps_2, \delta_2)$-DP algorithm and returning their results gives an algorithm that is $(\eps_1+\eps_2, \delta_1+\delta_2)$-DP. While  basic composition can be used to bound the DP properties of $k$ algorithms, it is known to not be tight, in particular for large values of $k$. In fact,  \emph{advanced} composition~\cite{dwork2010boosting} yields a general improvement, often translating to  $\approx \sqrt{k}$ reduction in the $\eps$ privacy bound of the composition of $k$ mechanisms each of which being $(\eps_0, \delta_0)$-DP.  Such a reduction can be sizeable in practical deployments, and therefore much research has been focusing on obtaining tighter composition bounds in various settings.
 
In the aforementioned setting where each mechanism has the same DP parameters, Kairouz et al.~\cite{kairouz2015composition} derived the optimal composition bound.
For the more general case of composing $k$ mechanisms whose privacy parameters are possibly different, i.e., the $i$th mechanism is guaranteed to be $(\eps_i, \delta_i)$-DP for some parameters $\eps_i, \delta_i$, computing the (exact) DP parameters of the composed mechanism is known to be \#P-complete \cite{murtagh2016complexity}.

While the results of~\cite{kairouz2015composition,murtagh2016complexity} provide a complete picture of privacy accounting when we assume only that the $i$th mechanism is $(\eps_i, \delta_i)$-DP, we can often arrive at tighter bounds when taking into account some additional information about the privacy loss of the mechanisms. For example, the Moments Accountant~\cite{abadi2016deep} and R\'{e}nyi DP~\cite{mironov2017renyi} methods keep track of (upper bounds on) the Renyi divergences of the output distributions on two adjacent databases; this allows one to compute upper bounds on the privacy parameters. These tools were originally introduced in the context of deep learning with DP (where the composition is over multiple iterations of the learning algorithm) in which they provide significant improvements over simply using the DP parameters of each mechanism. Other known tools that can also be used to upper-bound the privacy parameters of composed mechanisms include concentrated DP~\cite{dwork2016concentrated, bun2016concentrated} and its truncated variant~\cite{bun2018composable}. These methods are however all known not to be tight, and do not allow a high-accuracy estimation of the privacy parameters.

A numerical method for estimating the privacy parameters of a DP mechanism to an arbitrary accuracy, which has been the subject of several recent works starting with \cite{meiser2018tight, sommer2019privacy}, relies on the \emph{privacy loss distribution} (PLD). This is the probability mass function of the so-called privacy loss random variable in the case of discrete mechanisms, and its probability density function in the case of continuous mechanisms. From the PLD of a mechanism, one can easily obtain its (tight) privacy parameters. Moreover, a crucial property is that the PLD of a composition of multiple mechanisms is the convolution of their individual PLDs. Thus, \cite{koskela2020computing} used the Fast Fourier Transform (FFT) in order to speed up the computation of the PLD of the composition. Furthermore, explicit bounds on the approximation error for the resulting algorithm were derived in \cite{koskela2020computing,koskela2021tight,koskela2021computing,gopi2021numerical}. The PLD has been the basis of multiple open-source implementations from both industry and academia including \cite{DPBayes, GoogleDP, MicrosoftDP}. We note that the PLD can be applied to mechanisms whose privacy loss random variables do not have bounded moments, and thus for which composition cannot be analyzed using the Moments Accountant or R\'{e}nyi DP methods. An example such mechanism is DP-SGD-JL from \cite{bu2021fast}.

A crucial step in previous papers that use PLDs is in approximating the distribution so that it has finite support; this is especially needed in the case where the PLD is continuous or has a support of a very large size, as otherwise the FFT cannot be performed efficiently. With the exception of~\cite{gopi2021numerical}\footnote{Gopi et al.~\cite{gopi2021numerical} uses an estimator that is neither pessimistic nor optimistic, and instead derive their final values of $\delta$ using concentration bound-based error estimates.}, previous PLD-based accounting approaches~\cite{meiser2018tight,koskela2020computing,koskela2021tight,koskela2021computing} employ \emph{pessimistic estimators} and \emph{optimistic estimators} of PLDs. Roughly speaking, the former overestimate (i.e., give upper bounds on) $\delta$, whereas the latter underestimate $\delta$. For efficiency reasons, we would like the support of the approximate PLDs to be as small as possible, while retaining the accuracy of the estimates.

\subsubsection*{Our Contributions}

Our main contributions are the following:
\begin{itemize}
\item We obtain a new pessimistic estimator for a PLD and a given desired support set. Our pessimistic estimator is simple to construct and is based on the idea of ``connecting the dots'' of the hockey-stick curve at the discretization intervals.  Interestingly, we show that this is the best possible pessimistic estimator (and therefore is at least as good as previous estimators).  

\item We complement the above result by obtaining a new optimistic estimator that underestimates the PLD.  This estimator is based on the combination of a greedy algorithm and a convex hull computation.  In contrast to the pessimistic case, we prove that there is no ``best'' possible optimistic estimator.

\item We conduct an experimental evaluation showing that our estimators can work with much larger discretization intervals while keeping a similar error bound compared to previous approaches and yet give a better approximation than existing methods.

\end{itemize}

\section{Preliminaries}

For $k \in \N$, we use $[k]$ to denote $\{1, \dots, k\}$. For a set $S \subseteq \R \cup \{-\infty, +\infty\}$, we write $\exp(S)$ to denote $\{e^a \mid a \in S\}$. Similarly, for $S \subseteq \R_{\geq 0} \cup \{+ \infty\}$, we use $\log(S)$ to denote $\{\log(a) \mid a \in S\}$. Here we use the (standard) convention that $e^{+\infty} = +\infty$ and $e^{-\infty} = 0$; we also use the convention that $(+\infty) \cdot 0 = (-\infty) \cdot 0 = 0$. Moreover, we use $[x]_+$ as a shorthand for $\max\{x, 0\}$.

We use $\supp(P)$ to denote the support of a probability distribution $P$. For two distributions $P, Q$, we use $P \otimes Q$ to denote the product distribution of the two. Furthermore, when $P, Q$ are over an additive group, we use $P \ast Q$ to denote the convolution of the two distributions.

\subsection{Hockey-Stick Divergence and Curve} 

Let $\alpha \geq 0$.  
The {\em $\alpha$-hockey-stick divergence} between two probability distributions $P$ and $Q$ over a domain $\Omega$ is given as
\begin{equation}
D_{\alpha}(P || Q) := \sup_{S} \insquare{P(S) - \alpha \cdot Q(S)}_+, \label{eq:hockey-stick}
\end{equation}
where $\sup_S$ is over all measurable sets $S \subseteq \Omega$. 

For any pair $(A, B)$ of distributions, let $h_{(A, B)} : \bbR_{\ge 0} \cup \set{+ \infty} \to [0, 1]$ be its {\em hockey-stick curve}, given as $h_{(A, B)}(\alpha) := D_{\alpha}(A || B) = \sup_S [A(S) - \alpha \cdot B(S)]_+$.

The following characterization of hockey-stick curves, due to~\cite{zhu21optimal}, is helpful:

\begin{lemma}[\cite{zhu21optimal}] \label{lem:hockey-stick-characterization}
A function $h: \bbR_{\ge 0} \cup \set{+ \infty} \to [0, 1]$ is a hockey-stick curve for some pair of distributions if and only if the following three conditions hold:
\begin{enumerate}[(i)]
	\item $h$ is convex and non-increasing,
	\item $h(0) = 1$,
	\item $h(\alpha) \geq [1 - \alpha]_+$ for all $\alpha \in \bbR_{\ge 0} \cup \set{+ \infty}$.
\end{enumerate} 
\end{lemma}

\subsection{Differential Privacy}

The definition of differential privacy (DP) \cite{dwork06calibrating,dwork2006our}\footnote{For more background on differential privacy, we refer the reader to the monograph \cite{DworkR14}.} can be stated in terms of the hockey-stick divergence as follows.
\begin{definition}
For a notion of {\em adjacent datasets}, a mechanism $\calM$ is said to satisfy \emph{$(\eps, \delta)$-differential privacy} (denoted, $(\eps, \delta)$-DP) if for all adjacent datasets $S \sim S'$, it holds that $D_{e^{\eps}}(\calM(S) || \calM(S')) \le \delta$. 
\end{definition}

We point out that the techniques developed in this paper are general and do not depend on the specific adjacency relation. For the rest of the paper, for convenience, we always use $\alpha$ to denote $e^\eps$.

In most situations however, mechanisms satisfy $(\eps,\delta)$-DP for multiple values of $\eps$ and $\delta$. This is captured by the {\em privacy loss profile} $\delta_{\calM} : \bbR \to \bbR$ of a mechanism $\calM$ given as $\delta_{\calM}(\eps) := \sup_{S \sim S'} \calD_{e^\eps}(\calM(S) || \calM(S'))$.
It will be more convenient to consider the hockey-stick curve instead of the privacy profile. The only difference is that the hockey-stick curve takes $\alpha = e^\eps$ as parameter instead of $\eps$ as in the privacy profile.  %

\subsection{Dominating Pairs}

A central notion in our work is that of a {\em dominating pair} for a mechanism, defined by Zhu et al. \cite{zhu21optimal}.

\begin{definition}[Dominating Pairs~\cite{zhu21optimal}]
A pair $(P, Q)$ of distributions {\em dominates} a pair $(A, B)$ of distributions if it holds that
\[
\forall \alpha \ge 0 \ : \ D_{\alpha}(A || B) ~\le~ D_{\alpha}(P || Q);
\]
we denote this as $(A, B) \preceq (P, Q)$.  

A pair $(P, Q)$ of distributions is a {\em dominating pair for a mechanism} $\calM$ if for all adjacent datasets $S \sim S'$, it holds that $(\calM(S), \calM(S')) \preceq (P, Q)$; we denote this as $\calM \preceq (P, Q)$.

A pair $(P, Q)$ of distributions is a {\em tightly dominating pair} for $\calM$ if for every $\alpha \ge 0$, it holds that $D_{\alpha}(P || Q) = \sup_{S \sim S'}D_{\alpha}(\calM(S) || \calM(S'))$.
\end{definition}

Note that, by definition, $(P, Q) \succeq (A, B)$ if and only if $h_{(P, Q)}$ is no smaller than $h_{(A, B)}$ pointwise, i.e., $h_{(P, Q)}(\alpha) \geq h_{(A, B)}(\alpha)$ for all $\alpha \in \bbR_{\ge 0} \cup \set{+ \infty}$.

The following result highlights the importance of {\em dominating pairs}.
\begin{theorem}[\cite{zhu21optimal}] \label{thm:dom-pair}
If $\calM \preceq (P, Q)$ and $\calM' \preceq (P', Q')$, then $\calM \circ \calM' \preceq (P \otimes P', Q \otimes Q')$, where $\calM \circ \calM'$ is the composition of $\calM$ and $\calM'$. Furthermore, this holds even for {\em adaptive} composition.%
\footnote{In adaptive composition of $\calM' \circ \calM$, $\calM'$ can also take the output of $\cM$ as an auxiliary input. Here the $\cM' \preceq (P', Q')$ has to hold for all possible auxiliary input.}
\end{theorem}

Thus, in order to upper bound the privacy loss profile $\delta_{\calM}(\eps) := \sup_{S \sim S'} D_{e^{\eps}}(\calM(S) || \calM(S'))$, it suffices to compute $D_{e^{\eps}}(P || Q)$ for a dominating $(P, Q)$ pair for $\calM$. %

\subsection{Privacy Loss Distribution}

Privacy Loss Distribution (PLD)~\cite{dwork2016concentrated,sommer2019privacy} is yet another way to represent the privacy loss. For simplicity, we give a definition below specific to discrete distributions $P, Q$; it can be extended, e.g., to continuous distributions by replacing the probability masses $P(o), Q(o)$ with probability densities of $P, Q$ at $o$.

\begin{definition}[\cite{dwork2016concentrated}]
The \emph{privacy loss distribution} (PLD) of a pair  $(P, Q)$ of discrete distributions, denoted by $\PLD_{(P, Q)}$, is the distribution of the \emph{privacy loss random variable} $L$ generated by drawing $o \sim P$ and let $L = P(o)/Q(o)$.
\end{definition} 

As alluded to earlier, PLD can be used to compute the hockey-stick divergence~\cite{meiser2018tight,sommer2019privacy} (proof provided in \Cref{apx:proofs} for completeness):
\begin{restatable}{lemma}{lemPldHockeyStick} \label{lem:pld-to-hockey-stick}
For any pair $(P, Q)$ of discrete distributions and $\eps \in \R \cup \{-\infty, +\infty\}$, we have
\begin{align*}
	D_{e^\eps}(P||Q) := \sum_{\eps' \in \supp(\PLD_{(P, Q)})} [1 - e^{\eps - \eps'}]_+ \cdot \PLD_{(P, Q)}(\eps').
\end{align*}
\end{restatable}
Note that the RHS term above depends only on $\PLD_{(P, Q)}$ and not directly on $P, Q$ themselves. For convenience, we will abbreviate the RHS term as $D_{e^\eps}(\PLD_{(P, Q)})$.

The main advantage in dealing with PLDs is that composition simply corresponds to convolution of PLDs~\cite{meiser2018tight,sommer2019privacy}:
\begin{lemma} \label{lem:pld-composition}
Let $P, Q, P', Q'$ be discrete distributions. Then we have
\begin{align*}
	\PLD_{(P \otimes P', Q \otimes Q')} = \PLD_{(P, Q)} \ast \PLD_{(P', Q')}. 
\end{align*}
\end{lemma}

\subsection{Accounting Framework via Dominating Pairs and PLDs}
\label{sec:accounting-framework}

Dominating pairs and PLDs form a powerful set of building blocks to perform privacy accounting. Recall that in privacy accounting, we typically have a mechanism $\cM = \cM_1 \circ \cdots \circ \cM_k$ where each $\cM_i$ is a ``simple'' mechanism (e.g., Laplace or Gaussian mechanisms) and we would like to understand the privacy profile of $\cM$.

The approach taken in previous works~\cite{meiser2018tight,koskela2020computing,koskela2021tight,koskela2021computing} can be summarized as follows.%
\footnote{Note that their results are not phrased in terms of dominating pairs, since the latter is only defined and studied in~\cite{zhu21optimal}. Nonetheless, these previous works use similar (but more restricted) notions for ``worst case'' distributions.}
\begin{enumerate}
\item Identify a dominating pair $(A_i, B_i)$ for each $\cM_i$.
\item Find a \emph{pessimistic estimate}\footnote{We remark that this is slightly inaccurate as the ``pessimistic estimate'' in previous works may actually not be a valid PLD; please see \Cref{sec:pb-pessimistic} for a more detailed explanation.} $(P^{\uparrow}_i, Q^{\uparrow}_i) \succeq (A_i, B_i)$ such that $\PLD_{(P^{\uparrow}_i, Q^{\uparrow}_i)}$ is supported on a certain set of prespecified values.
\item Compute $\PLD^{\uparrow} = \PLD_{(P^{\uparrow}_1, Q^{\uparrow}_1)} \ast \cdots \ast \PLD_{(P^{\uparrow}_k, Q^{\uparrow}_k)}$.
\item Compute $\delta^{\uparrow}(\eps)$ from $\PLD^{\uparrow}$ using the formula from \Cref{lem:pld-to-hockey-stick}.
\end{enumerate}

By \Cref{thm:dom-pair} and \Cref{lem:pld-composition,lem:pld-to-hockey-stick}, we can conclude that $\delta^{\uparrow}(\eps) \geq \delta_\cM(\eps)$; in other words, the mechanism $\cM$ is $(\eps, \delta^{\uparrow}(\eps))$-DP as desired.

Note that the reason that one needs $\PLD_{(P^{\uparrow}_i, Q^{\uparrow}_i)}$ to have finite support in the second step is so that it can be computed efficiently via the Fast Fourier Transform (FFT). Currently, there is only one approach used in previous works, called \emph{Privacy Buckets}~\cite{meiser2018tight}. Roughly speaking, this amounts simply to rounding the PLD up to the nearest point in the specified support set. (See \Cref{sec:pb-pessimistic} for a more formal description.)

While the above method gives us an upper bound $\delta^{\uparrow}(\eps)$ of $\delta_\cM(\eps)$, there are scenarios where we would like to find a \emph{lower bound} on $\delta_\cM(\eps)$; for example, this can be helpful in determining how tight our upper bound is. Computing such a lower bound is also possible under the similar framework, except that we need to know a list of tightly dominating pairs $(A^*_1, B^*_1), \dots, (A^*_k, B^*_k)$ such that there exists an adjacent datasets $S, S'$ for which $D_{e^\eps}(\cM(S) || \cM(S')) = D_{e^\eps}(A^*_1 \otimes \cdots \otimes A^*_k || B^*_1 \otimes \cdots \otimes B^*_k)$. If such tightly dominating pairs can be identified, then we can follow the same blueprint as above except we replace a pessimistic estimate with an \emph{optimistic estimate} $(P^{\downarrow}, Q^{\downarrow}) \preceq (A^*_i, B^*_i)$. This would indeed gives us a lower bound $\delta^{\downarrow}(\eps)$ of $\delta_\cM(\eps)$.

The described framework is illustrated in \Cref{fig:accounting-framework}.

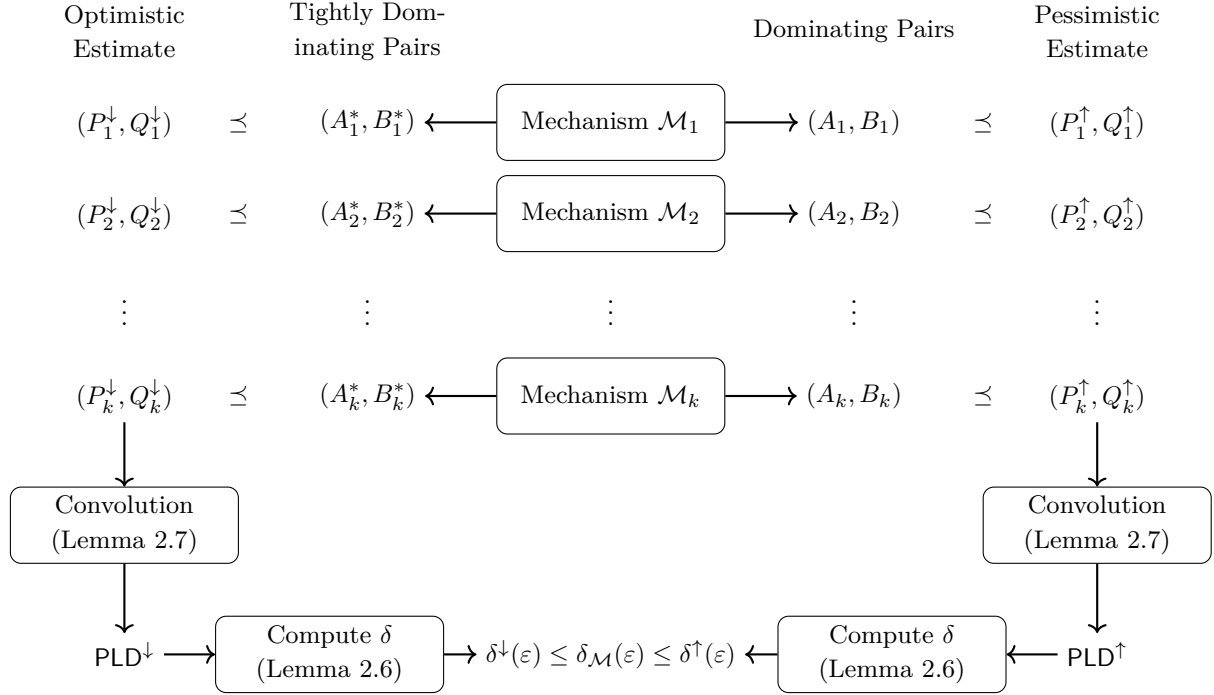
\begin{figure*}
\centering
\begin{tikzpicture}[node distance=1.2cm]
	\node (m1) [dnode] {Mechanism $\cM_1$};
	\node (m2) [dnode,below of=m1] {Mechanism $\cM_2$};
	\node (mdot) [below of=m2] {$\vdots$};
	\node (mk) [dnode,below of=mdot] {Mechanism $\cM_k$};
	\node(d1) [right of=m1,xshift=2cm] {$(A_1, B_1)$};
	\draw [arrow] (m1) -- (d1);
	\node(d2) [right of=m2,xshift=2cm] {$(A_2, B_2)$};
	\draw [arrow] (m2) -- (d2);
	\node(ddot) [right of=mdot,xshift=2cm] {$\vdots$};
	\node(dk) [right of=mk,xshift=2cm] {$(A_k, B_k)$};
	\draw [arrow] (mk) -- (dk);
	\node(dp) [above of=d1,text width=3cm,align=center] {Dominating Pairs};
	\node(p1) [right of=d1,xshift=0.5cm] {$\preceq$};
	\node(p2) [right of=d2,xshift=0.5cm] {$\preceq$};
	\node(pk) [right of=dk,xshift=0.5cm] {$\preceq$};
	\node(p1) [right of=d1,xshift=2cm] {$(P^{\uparrow}_1, Q^{\uparrow}_1)$};
	\node(p2) [right of=d2,xshift=2cm] {$(P^{\uparrow}_2, Q^{\uparrow}_2)$};
	\node(pdot) [right of=ddot,xshift=2cm] {$\vdots$};
	\node(pk) [right of=dk,xshift=2cm] {$(P^{\uparrow}_k, Q^{\uparrow}_k)$};
	\node(pconv) [dnode,below of=pk,text width=2cm,yshift=-0.5cm] {Convolution (\Cref{lem:pld-composition})};
	\draw [arrow] (pk) -- (pconv);
	\node (cpconv) [below of=pconv,yshift=-0.5cm] {$\PLD^{\uparrow}$};
	\draw [arrow] (pconv) -- (cpconv);
	\node (pcompdelta) [dnode,left of=cpconv,text width=2cm,xshift=-1.5cm] {Compute $\delta$ (\Cref{lem:pld-to-hockey-stick})};
	\draw [arrow] (cpconv) -- (pcompdelta);
	\node(pp) [above of=p1,text width=2cm,align=center] {Pessimistic Estimate};
	\node(td1) [left of=m1,xshift=-2cm] {$(A^*_1, B^*_1)$};
	\draw [arrow] (m1) -- (td1);
	\node(td2) [left of=m2,xshift=-2cm] {$(A^*_2, B^*_2)$};
	\draw [arrow] (m2) -- (td2);
	\node(tddot) [left of=mdot,xshift=-2cm] {$\vdots$};
	\node(tdk) [left of=mk,xshift=-2cm] {$(A^*_k, B^*_k)$};
	\draw [arrow] (mk) -- (tdk);
	\node(tdp) [above of=td1,text width=3cm,align=center] {Tightly Dominating Pairs};
	\node(o1) [left of=td1,xshift=-.5cm] {$\preceq$};
	\node(o2) [left of=td2,xshift=-.5cm] {$\preceq$};
	\node(ok) [left of=tdk,xshift=-.5cm] {$\preceq$};
	\node(o1) [left of=td1,xshift=-2cm] {$(P^{\downarrow}_1, Q^{\downarrow}_1)$};
	\node(o2) [left of=td2,xshift=-2cm] {$(P^{\downarrow}_2, Q^{\downarrow}_2)$};
	\node(odot) [left of=tddot,xshift=-2cm] {$\vdots$};
	\node(ok) [left of=tdk,xshift=-2cm] {$(P^{\downarrow}_k, Q^{\downarrow}_k)$};
	\node(oconv) [dnode,below of=ok,text width=2cm,yshift=-0.5cm] {Convolution (\Cref{lem:pld-composition})};
	\draw [arrow] (ok) -- (oconv);
	\node (coconv) [below of=oconv,yshift=-0.5cm] {$\PLD^{\downarrow}$};
	\draw [arrow] (oconv) -- (coconv);
	\node (ocompdelta) [dnode,right of=coconv,text width=2cm,xshift=1.5cm] {Compute $\delta$ (\Cref{lem:pld-to-hockey-stick})};
	\draw [arrow] (coconv) -- (ocompdelta);
	\node(op) [above of=o1,text width=2cm,align=center] {Optimistic Estimate};
	\node(delta) [below of=mk,yshift=-2.2cm] {$\delta^{\downarrow}(\eps) \leq \delta_{\cM}(\eps) \leq \delta^{\uparrow}(\eps)$};
	\draw [arrow] (pcompdelta) -- (delta);
	\draw [arrow] (ocompdelta) -- (delta);
\end{tikzpicture}
\caption{
	Illustration of the framework for privacy accounting using PLDs and the notion of (tightly) dominating pairs. 
	\label{fig:accounting-framework}%
}
\end{figure*}

\section{Finitely-Supported PLDs}

As alluded to in the previous section, to take full advantage of FFT, it is important that a PLD is discretized in  a way such that its support is finite. For exposition purposes, we will assume that the discretization points include $-\infty$ and $+\infty$.  We will use $\cE$ for discretization points for the PLD and $\cA$ for the corresponding discretization points for the hockey-stick curve:

\begin{assumption}
Let $\cA = \{\alpha_0, \dots, \alpha_k\}$ be any finite subset of $\bbR_{\ge 0} \cup \set{+ \infty}$ such that $0 = \alpha_0 < \alpha_1 < \cdots < \alpha_k = +\infty$, and let $\cE = \{\eps_0, \dots, \eps_k\}$ be such that $\eps_0 = -\infty, \eps_k = +\infty$ and $\eps_i = \log(\alpha_i)$ for all $i \in [k - 1]$.
\end{assumption}

For the remaining of this work, we will operate under the above assumption and we will not state this explicitly for brevity.

Using the characterization in \Cref{lem:hockey-stick-characterization}, we can also characterize the hockey-stick curve of PLDs whose support is on a prespecified finite set $\cE$, stated more precisely in the lemma below. Furthermore, the ``inverse'' part of the lemma yields an algorithm (\Cref{alg:discretize}) that can construct $A, B$ given $(h(\alpha_i))_{\alpha_i \in \cA}$ which we will use in the sequel.

\begin{lemma} \label{lem:hockey-stick-characterization-finite-support}
A function $h: \bbR_{\ge 0} \cup \set{+ \infty} \to [0, 1]$ is a hockey-stick curve for some pair $(P, Q)$ such that $\supp(\PLD_{(P, Q)}) \subseteq \cE$ if and only if the following conditions hold:
\begin{enumerate}[(i)]
	\item $h$ is convex and non-increasing,
	\item $h(0) = 1$,
	\item $h(\alpha_i) \geq [1 - \alpha_i]_+$ for all $\alpha_i \in \cA$,
	\item For all $i \in [k - 1]$, the curve $h$ restricted to $[\alpha_{i - 1}, \alpha_i]$ is linear: i.e., for all $\alpha \in [\alpha_{i - 1}, \alpha_i)$, we have $h(\alpha) = \frac{\alpha - \alpha_i}{\alpha_{i - 1} - \alpha_i} \cdot h(\alpha_{i - 1}) + \frac{\alpha_{i - 1} - \alpha}{\alpha_{i - 1} - \alpha_i} \cdot h(\alpha_i)$.
	\item For all $\alpha > \alpha_{k - 1}$, $h(\alpha) = h(+\infty)$.
\end{enumerate} 
\end{lemma}

A consequence of \Cref{lem:hockey-stick-characterization-finite-support} is that $h$ is completely specified by $h(\cA)$. More formally, given $f: \cA \to [0, 1]$, the only possible extension of $f$ to a hockey-stick curve is its \emph{piecewise-linear} extension $\of$ defined by
\begin{align*}
&\of(\alpha) := \\
&\begin{cases}
	\frac{\alpha - \alpha_i}{\alpha_{i - 1} - \alpha_i} f(\alpha_{i - 1}) + \frac{\alpha_{i - 1} - \alpha}{\alpha_{i - 1} - \alpha_i} f(\alpha_i) &\text{if } \alpha \in [\alpha_{i - 1}, \alpha_i) \\
	f(+\infty) &\text{if } \alpha > \alpha_{k - 1},
\end{cases}
\end{align*}
for all $\alpha \in \R_{\geq 0} \cup \{+\infty\}$. Note that this $\of$ may still not be a hockey-stick curve, as it may not be convex.

\begin{proof}[Proof of \Cref{lem:hockey-stick-characterization-finite-support}]
$(\Leftarrow)$ We start with the converse direction by describing an algorithm that, given $h(\alpha_0), \dots, h(\alpha_k)$, can construct the desired $P, Q$. In fact, we will construct distributions $P$ and $Q$ with supports contained in $\cA$ satisfying the following:
\begin{enumerate}[($\Pi_1$),leftmargin=12mm]
	\item $P(\alpha) = \alpha \cdot Q(\alpha)$ for all $\alpha \in \calA \setminus \{+\infty\}$, \label{property:p1}
	\item $Q(\infty) = 0$, \label{property:p2}
	\item $D_{\alpha}(P || Q) = h(\alpha)$ for all $\alpha \in \calA$. \label{property:p3}
\end{enumerate}
The first two conditions imply that $\supp(\PLD_{(P,Q)}) \subseteq \cE$ and the last condition implies that $h_{(P, Q)} = h$ as desired.

The construction of $P, Q$ is described in \Cref{alg:discretize}. 

\begin{algorithm}
	\caption{PLD Discretization.} \label{alg:discretize}
	\begin{algorithmic}
		\Procedure{DiscretizePLD}{$h(\alpha_0), \dots, h(\alpha_k)$}
		\State $Q(\alpha_k) \gets 0$
		\Comment{$\alpha_k = +\infty$}
		\For{$i = k-1, \dots, 1$}
		\State $Q(\alpha_i) \gets \frac{h(\alpha_{i-1}) - h(\alpha_i)}{\alpha_{i} - \alpha_{i-1}} - \frac{h(\alpha_i) - h(\alpha_{i+1})}{\alpha_{i+1} - \alpha_i}$
		\EndFor
		\State $Q(\alpha_0) \leftarrow 1 - \sum_{j \in [k - 1]} Q(\alpha_j)$
		\Comment{$\alpha_0 = 0$}
		\State $P(\alpha_0) \leftarrow 0$ \Comment{$\alpha_0 = 0$}
		\For{$i = 1, \dots, k - 1$}
		\State $P(\alpha_i) \leftarrow \alpha_i \cdot Q(\alpha_i)$ 
		\EndFor
		\State $P(\alpha_k) \gets h(\alpha_k)$
		\Comment{$\alpha_k = +\infty$}
		
		\EndProcedure
	\end{algorithmic}
\end{algorithm}

Let us now verify that both $P$ and $Q$ are valid probability distributions. First, notice that $Q(\alpha_i) \geq 0$ due to the convexity of $h$. Furthermore, 
\begin{align*}
	Q(0) = 1 - \sum_{j \in [k - 1]} Q(\alpha_j) = 1 - \frac{1 - h(\alpha_1)}{\alpha_1} \geq 0,
\end{align*}
where the last inequality follows from (iii). Thus, $Q$ is indeed a probability distribution. As for $P$, notice that $P(\alpha_i) \geq 0$ for all $\alpha_i \in \cA$. Finally, we also have
\begin{align*}
	&\sum_{i \in \{0, \dots, k\}} P(\alpha_i) \\
	&= h(+\infty) + \sum_{i \in \{0, \dots, k - 1\}} \alpha_{i} \cdot Q(\alpha_{i}) \\
	&= h(+\infty) + \\ \quad &\sum_{i \in \{0, \dots, k - 1\}} \alpha_{i} \left(\frac{h(\alpha_{i-1}) - h(\alpha_i)}{\alpha_{i} - \alpha_{i-1}} - \frac{h(\alpha_i) - h(\alpha_{i+1})}{\alpha_{i+1} - \alpha_i}\right) \\
	&= h(+\infty) + \sum_{i \in [k - 1]} (\alpha_{i + 1} - \alpha_i) \cdot \frac{h(\alpha_i) - h(\alpha_{i+1})}{\alpha_{i+1} - \alpha_i} \\
	&= h(+\infty) + (h(0) - h(+\infty)) \\
	&= 1,
\end{align*}
meaning that $P$ is a probability distribution as desired.

Properties~\ref{property:p1} and~\ref{property:p2} are immediate from the construction. We will now check Property~\ref{property:p3}, based on two cases whether $\alpha > \alpha_{k - 1}$.
\begin{itemize}
	\item Case I: $\alpha \geq \alpha_{k - 1}$. In this case, we have
	$D_{\alpha}(P||Q) = P(+\infty) - e^\eps Q(+\infty) = h(+\infty)$.
	\item Case II: $\alpha < \alpha_{k - 1}$. Suppose that $\alpha \in [\alpha_{i - 1}, \alpha_i)$ for $i \in [k - 1]$. We have
	\begin{align*}
		D_{\alpha}(P||Q) &= P(\{\alpha_i, \dots, \alpha_k\}) - \alpha \cdot Q(\{\alpha_i, \dots, \alpha_k\}) \\
		&= \sum_{j=i}^{k} (\alpha_j - \alpha) \cdot Q(\alpha_j) \\
		&= \frac{\alpha - \alpha_i}{\alpha_{i - 1} - \alpha_i} \cdot \sum_{j=i}^{k} (\alpha_j - \alpha_{i - 1}) \cdot Q(\alpha_j) \\ &\quad+ \frac{\alpha_{i - 1} - \alpha}{\alpha_{i - 1} - \alpha_i} \cdot \sum_{j=i}^{k} (\alpha_j - \alpha_i) \cdot Q(\alpha_j).
	\end{align*}
	Furthermore, we have
	\begin{align*}
		&\sum_{j=i}^{k} (\alpha_j - \alpha_{i - 1}) \cdot Q(\alpha_j) \\
		&=\sum_{j=i}^{k} (\alpha_j - \alpha_{i - 1}) \\&\qquad \cdot\left(\frac{h(\alpha_{j-1}) - h(\alpha_j)}{\alpha_{j} - \alpha_{j-1}} - \frac{h(\alpha_j) - h(\alpha_{j+1})}{\alpha_{j+1} - \alpha_j}\right) \\
		&= h(\alpha_{i - 1}).
	\end{align*}
	Similarly, we also have $\sum_{j=i}^{k} (\alpha_j - \alpha_i) \cdot Q(\alpha_j) = h(\alpha_i)$.
	Combining the three equalities, we arrive at
	\begin{align*}
		D_{\alpha}(P||Q) &= \frac{\alpha - \alpha_i}{\alpha_{i - 1} - \alpha_i} \cdot h(\alpha_{i - 1}) + \frac{\alpha_{i - 1} - \alpha}{\alpha_{i - 1} - \alpha_i} \cdot h(\alpha_i),
	\end{align*}
	which is equal to $h(\alpha)$ due to assumption (iv).
\end{itemize}
As a result, $h_{(P, Q)} = h$ as desired.

$(\Rightarrow)$ We will now prove this direction. (i), (ii), and (iii) follow immediately from~\Cref{lem:hockey-stick-characterization}. As a result, it suffices to only prove (iv) and (v). Suppose that $h_{(P, Q)} = h$ for some pair $(P, Q)$ such that $\supp(\PLD_{(P,Q)}) \subseteq \cE$. Let $R$ be a shorthand for the distribution of $\exp(\PLD_{(P,Q)})$. To prove (iv), consider any $\alpha \in [\alpha_{i-1}, \alpha_i)$ for some $i \in [k - 1]$. We then have
\begin{align*}
	h(\alpha) &= D_{\alpha}(P||Q) \\
	&= \sum_{j=i}^{k} (1 - \alpha/\alpha_j) \cdot R(\alpha_j) \\
	&= \frac{\alpha - \alpha_i}{\alpha_{i - 1} - \alpha_i} \cdot \sum_{j=i}^{k} (1 - \alpha_{i-1}/\alpha_j) \cdot R(\alpha_j) \\
	&\qquad \frac{\alpha_{i - 1} - \alpha}{\alpha_{i - 1} - \alpha_i} \cdot \sum_{j=i}^{k} (1 - \alpha_i/\alpha_j) \cdot R(\alpha_j) \\
	&=  \frac{\alpha - \alpha_i}{\alpha_{i - 1} - \alpha_i} \cdot h(\alpha_{i - 1}) + \frac{\alpha_{i - 1} - \alpha}{\alpha_{i - 1} - \alpha_i} \cdot h(\alpha_i),
\end{align*}
which completes the proof of (iv).

Next, we prove (v). Consider any $\alpha \geq \alpha_{k - 1}$. We have
\begin{align*}
	h(\alpha) = D_{\alpha}(P||Q) = R(+\infty) = h(+\infty), \\
\end{align*}  
thereby completing our proof.
\end{proof}

\section{Pessimistic PLDs with Finite Support}
\label{sec:pessimistic}

\begin{figure*}
	\centering
	\begin{subfigure}[b]{0.48\textwidth}
		\centering
		\includegraphics[width=\textwidth]{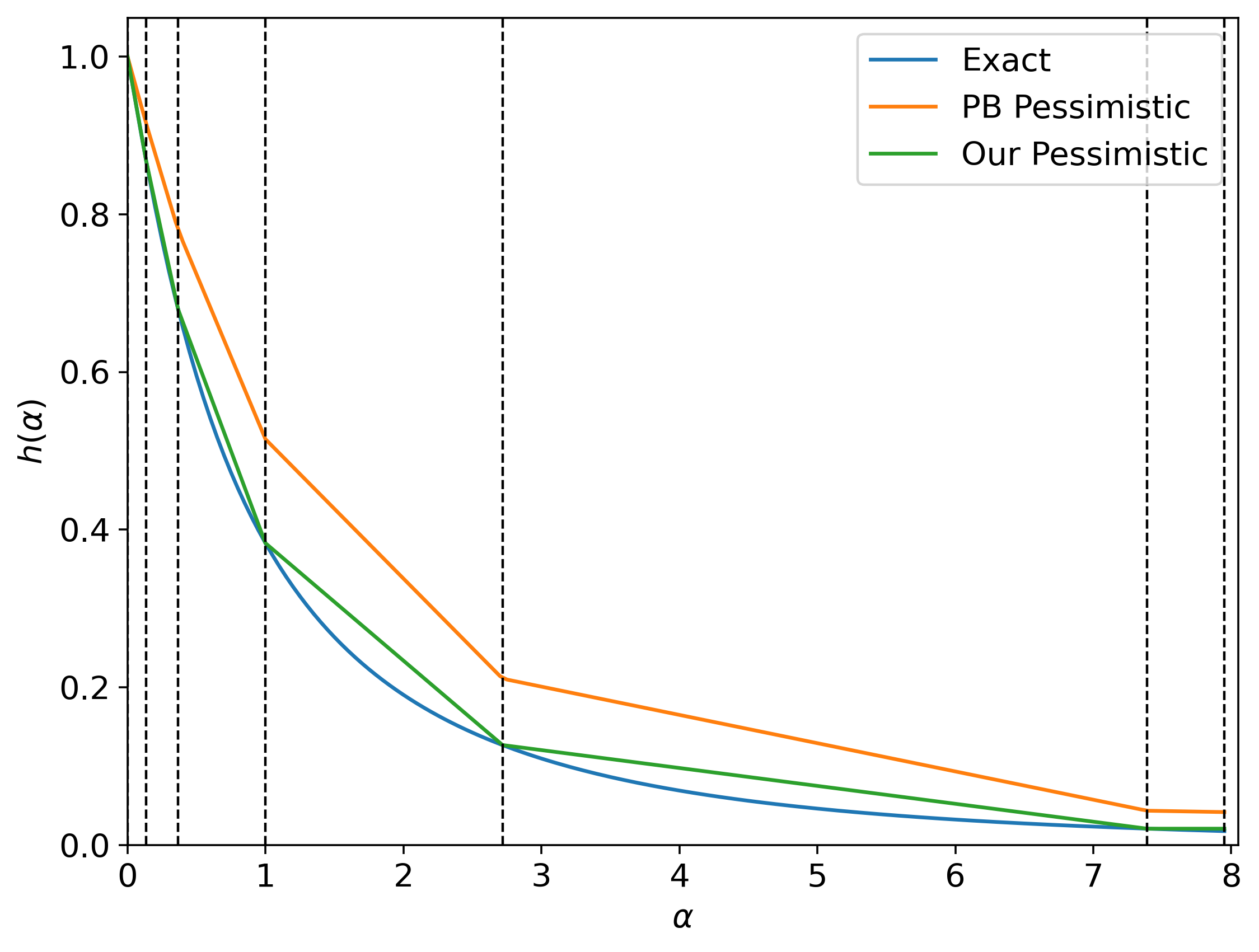}
		\caption{Gaussian Mechanism}
		\label{fig:pess-gaussian}
	\end{subfigure}
	\hfill
	\begin{subfigure}[b]{0.48\textwidth}
		\centering
		\includegraphics[width=\textwidth]{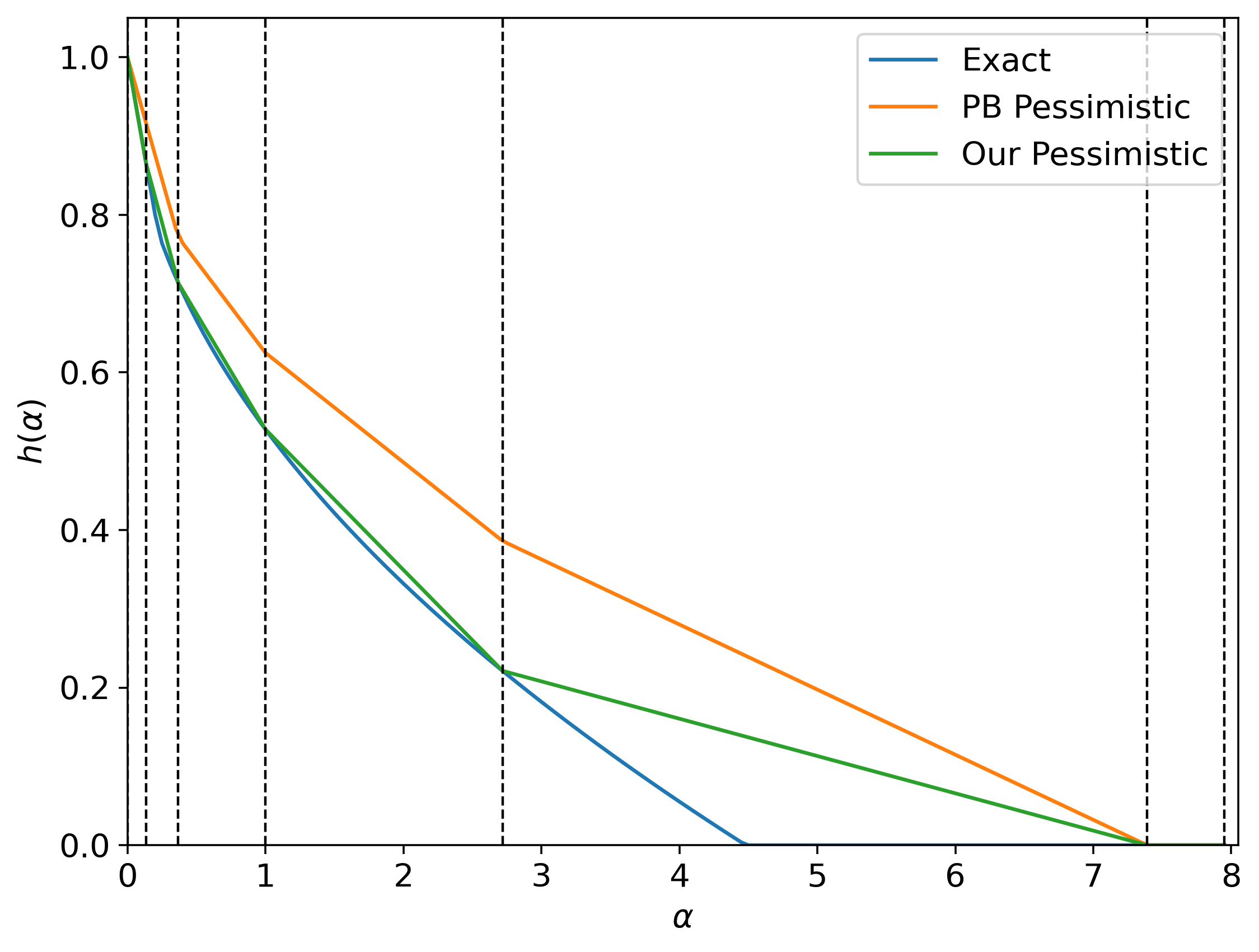}
		\caption{Laplace Mechanism}
		\label{fig:pess-laplace}
	\end{subfigure}
	\caption{Illustrations of the hockey-stick curves of the Gaussian and Laplace mechanisms (with noise multipliers equal to 1 and $2/3$ respectively), and their pessimistic estimates from our approach (labelled ``pessimistic'') and the Privacy Bucket (PB) approach (labelled ``PB pessimistic'') of~\cite{meiser2018tight}. The horizontal lines represent the discretization points in the set $\cA$. As corroborated by \Cref{lem:pess-better-than-pb}, our pessimistic estimates is closer to the true curves (labelled ``exact'') compared to the PB pessimistic estimates for all $\alpha$.}
	\label{fig:pess-simple}
\end{figure*}

As we have described in \Cref{fig:accounting-framework}, pessimistic estimates of PLDs with finite supports are crucial in the PLD-based privacy accounting framework.  The better these pessimistic estimates approximate the true PLD, the more accurate is the resulting upper bound $\delta^{\uparrow}(\eps)$. 

Equipped with tools developed in the previous section, we will now describe our finite-support pessimistic estimate of a PLD. Specifically, given a pair $(A, B)$ of distributions, we would like to compute $(P^{\uparrow}, Q^{\uparrow})$ such that $(P^{\uparrow}, Q^{\uparrow}) \succeq (A, B)$ with $\supp(\PLD_{(P^{\uparrow}, Q^{\uparrow})}) \subseteq \cE$. In fact, as we will show below (\Cref{lem:pessimistic-best}), our choice of $\PLD_{(P^{\uparrow}, Q^{\uparrow})}$ ``best approximates'' $\PLD_{(A, B)}$.

Our construction of the pair $(P^{\uparrow}, Q^{\uparrow})$ is  simple: run DiscretizePLD (\Cref{alg:discretize}) on the input $h_{(A, B)}(\alpha_0), \dots, h_{(A, B)}(\alpha_k)$. 

Recall from the proof of \Cref{lem:hockey-stick-characterization-finite-support} that this construction simply gives $h_{(P^{\uparrow}, Q^{\uparrow})}$, which is a piecewise-linear extension of $h_{(A, B)}(\alpha_0), \dots, h_{(A, B)}(\alpha_k)$. In other words, we simply ``connect the dots'' to construct the hockey-stick curve of our pessimistic estimate. Note that this, together with the convexity of $h_{(A, B)}$ (\Cref{lem:hockey-stick-characterization}), implies that $(P^{\uparrow}, Q^{\uparrow}) \succeq (A, B)$ as desired.

Additionally, it is not hard to observe that our choice of pessimistic PLD is the best possible, in sense that $(P^{\uparrow}, Q^{\uparrow})$ is the least element (under the domination partial order) among all pairs that dominate $(A, B)$:

\begin{lemma} \label{lem:pessimistic-best}
	Let $P, Q$ be any pair of distributions such that $\PLD_{(P, Q)}$ is supported on $\cA$ and $(P, Q) \succeq (A, B)$. Then, we must have $(P, Q) \succeq (P^{\uparrow}, Q^{\uparrow})$
\end{lemma}

\begin{proof}
	Recall that it suffices to prove that $h_{(P, Q)}(\alpha) \geq h_{(P^{\uparrow}, Q^{\uparrow})}(\alpha)$ for all $\alpha \in \bbR_{\ge 0} \cup \set{+ \infty}$. We will consider two cases based on the value of $\alpha$:
	\begin{itemize}
		\item Case I: $\alpha \geq \alpha_{k - 1}$. From \Cref{lem:hockey-stick-characterization-finite-support}, we simply have $h_{(P, Q)}(\alpha) = h_{(P, Q)}(+\infty) \geq h_{(A, B)}(+\infty) = h_{(P^{\uparrow}, Q^{\uparrow})}(\alpha)$.
		\item Case II: $\alpha_{k - 1} > \alpha \geq 0$. Suppose that $\alpha \in [\alpha_{i - 1}, \alpha_i)$. From \Cref{lem:hockey-stick-characterization-finite-support}(ii), we have 
		\begin{align*}
			&h_{(P, Q)}(\alpha) \\
			&= \frac{\alpha - \alpha_i}{\alpha_{i - 1} - \alpha_i} \cdot h_{(P, Q)}(\alpha_{i - 1}) + \frac{\alpha_{i - 1} - \alpha}{\alpha_{i - 1} - \alpha_i} \cdot h_{(P, Q)}(\alpha_i) \\
			&\geq \frac{\alpha - \alpha_i}{\alpha_{i - 1} - \alpha_i} \cdot h(\alpha_{i - 1}) + \frac{\alpha_{i - 1} - \alpha}{\alpha_{i - 1} - \alpha_i} \cdot h(\alpha_i) \\
			&=  h_{(P^{\uparrow}, Q^{\uparrow})}(\alpha),
		\end{align*}
		where the first inequality follows from $(P, Q) \succeq (A, B)$ and the last equality follows from our construction of $(P^{\uparrow}, Q^{\uparrow})$. \qedhere
	\end{itemize}
\end{proof}

We remark that $\PLD_{(P^{\uparrow}, Q^{\uparrow})}$ also has a simple form, due to the properties \ref{property:p1} and \ref{property:p2}:
\begin{align*}
	\PLD_{(P^{\uparrow}, Q^{\uparrow})}(\eps_i) := P^{\uparrow}(\alpha_i)
\end{align*}
for all $i \in [k]$.

\subsection{Comparison to Privacy Loss Buckets}
\label{sec:pb-pessimistic}

The primary previous work that also derived a pessimistic estimate of PLD is that of Meiser and Mohammadi~\cite{meiser2018tight}, which has also been used (implicitly) in later works~\cite{koskela2020computing,koskela2021tight,koskela2021computing}. In our terminology, the \emph{Privacy Buckets (PB)} algorithm of Meiser and Mohammadi~\cite{meiser2018tight}\footnote{This is referred to as \emph{grid approximation} in~\cite{koskela2020computing,koskela2021tight,koskela2021computing}.} can be restated as follows: let the pessimistic-PB estimate $\tPLD^{\uparrow}_{(A,B)}$ be the probability distribution where
\begin{align*}
	\tPLD^{\uparrow}_{(A,B)}(\eps_i) = \PLD_{(A,B)}((\eps_{i - 1}, \eps_i]),
\end{align*}
for all $i \in [k]$. In other words, $\tPLD^{\uparrow}_{(A,B)}$ is a probability distribution on $\cE$ that stochastically dominates $\PLD_{(A,B)}$; furthermore, $\tPLD^{\uparrow}_{(A,B)}$ is the least such distribution under stochastic dominant (partial) ordering. In previous works~\cite{meiser2018tight,koskela2020computing,koskela2021tight,koskela2021computing}, such an estimate is then used in place of the true (non-discretized) PLD for accounting and computing $\delta$'s (via \Cref{lem:pld-composition} and \Cref{lem:pld-to-hockey-stick}).

A priori, it is not clear whether $\tPLD^{\uparrow}_{(A,B)}$ is even a valid PLD (for some pair of distributions). However, it not hard to prove that this is indeed the case:
\begin{lemma}
	There exists a pair $(P_{\mathrm{PB}}^{\uparrow}, Q_{\mathrm{PB}}^{\uparrow})$ of distributions  such that $\tPLD^{\uparrow}_{(A,B)} = \PLD_{(P_{\mathrm{PB}}^{\uparrow}, Q_{\mathrm{PB}}^{\uparrow})}$.
\end{lemma}

\begin{proof}
	Let $P_{\text{PB}}^{\uparrow}$ be defined by
	\begin{align*}
		P_{\text{PB}}^{\uparrow}(\alpha_i) = \tPLD^{\uparrow}_{(A,B)}(\eps_i)
	\end{align*}
	for all $i \in \{0, \dots, k\}$. It is clear that $P_{\text{PB}}^{\uparrow}$ is a valid distribution.
	
	Then, define $Q_{\text{PB}}^{\uparrow}$ by
	\begin{align*}
		Q_{\text{PB}}^{\uparrow}(\alpha) = P_{\text{PB}}^{\uparrow}(\alpha) / \alpha,
	\end{align*}
	for all $\alpha \in \cA \setminus \{0\}$ and let $Q_{\text{PB}}^{\uparrow}(0) = 1 - \sum_{\alpha \in \cA \setminus \{0\}} Q_{\text{PB}}^{\uparrow}(\alpha)$. To check that $Q_{\text{PB}}^{\uparrow}$ is a valid distribution, it suffices to show that $Q_{\text{PB}}^{\uparrow}(0) \geq 0$. This is true because
	\begin{align*}
		\sum_{\alpha \in \cA \setminus \{0\}} Q_{\text{PB}}^{\uparrow}(\alpha) 
		&~=~ \sum_{i \in [k]} P_{\text{PB}}^{\uparrow}(\alpha_i) / \alpha_i \\
		&~=~ \sum_{i \in [k]} \tPLD^{\uparrow}_{(A,B)}(\eps_i) / \alpha_i \\
		&~=~ \sum_{i \in [k]} \PLD_{(A,B)}((\eps_{i - 1}, \eps_i]) / \alpha_i \\
		&~=~ \sum_{i \in [k]} \sum_{\substack{o \in \supp(B) \\ A(o) / B(o) \in (\alpha_{i - 1}, \alpha_i]}} A(o) / \alpha_i \\
		&~\leq~ \sum_{i \in [k]} \sum_{\substack{o \in \supp(B) \\ A(o) / B(o) \in (\alpha_{i - 1}, \alpha_i]}} B(o) \\
		&~\leq~ 1.
	\end{align*}
	
	Finally, it is obvious from the definitions of $P_{\text{PB}}^{\uparrow}, Q_{\text{PB}}^{\uparrow}$ that $\PLD_{(P_{\text{PB}}^{\uparrow}, Q_{\text{PB}}^{\uparrow})} = \tPLD^{\uparrow}_{(A,B)}$.
\end{proof}

Combining the above lemma and the fact that $\tPLD^{\uparrow}_{(A,B)}$ dominates $\PLD_{(A, B)}$ with \Cref{lem:pessimistic-best}, we can conclude that our estimate is no worse than the PB estimate:
\begin{corollary} \label{lem:pess-better-than-pb}
	Let $(P^{\uparrow}, Q^{\uparrow})$ be as defined above. Then, for all $\alpha \geq 0$, we have $h_{(P^{\uparrow}, Q^{\uparrow})}(\alpha) \leq D_\alpha(\tPLD^{\uparrow}_{(A,B)})$.
\end{corollary}
\iffalse
\begin{proof}
	Since $\supp(\tPLD^{\uparrow}_{(A, B)}) \subseteq \cE$, it is simple to see that $D_\alpha(\tPLD^{\uparrow}_{(A,B)})$ as a function of $\alpha$ is also piecewise-linear with breakpoints in $\cA$ (similar to (iv) and (v) of~\Cref{lem:hockey-stick-characterization-finite-support}). Therefore, it suffices to show that $h_{(P^{\uparrow}, Q^{\uparrow})}(\alpha_i) \leq D_{\alpha_i}(\tPLD^{\uparrow}_{(A,B)})$ just for all $\alpha_i \in \cA$. To see this, notice that
	\begin{align*}
		h_{(P^{\uparrow}, Q^{\uparrow})}(\alpha)
		= h_{(A,B)}(\alpha_i)
		&= D_\alpha(\PLD_{(A,B)}) \\
		&\leq D_\alpha(\tPLD^{\uparrow}_{(A,B)}),
	\end{align*} 
	where the first equality comes from our choice of $(P^{\uparrow}, Q^{\uparrow})$, the second one follows from \Cref{lem:pld-to-hockey-stick}, and the inequality follows from  stochastic dominance of $\tPLD^{\uparrow}_{(A,B)}$ over $\PLD_{(A,B)}$.
\end{proof}
\fi

Illustrations of the exact hockey-stick divergence and its pessimistic estimates from our approach and PB approach can be found in \Cref{fig:pess-simple}. A more detailed evaluation of the error from the two approaches (after compositions) can be found in \Cref{sec:exp}.

\section{Optimistic PLDs with Finite Support}
\label{sec:optimistic}

We next consider \emph{optimistic} PLDs, i.e., $\PLD_{(P,Q)}$ dominated by a given $\
\PLD_{(A,B)}$. We start by showing that, unlike pessimistic PLDs for which there is the ``best'' possible choice (\Cref{lem:pessimistic-best}), there is no such a choice for optimistic PLDs:

\begin{lemma} \label{lem:opt-not-unique}
	There exists a pair $(A, B)$ of distributions and a finite set $\cA$ such that, for any pair $(P, Q) \preceq (A, B)$ such that $\supp(\PLD_{(P, Q)}) \subseteq \cE$, there exists a pair $(P', Q') \preceq (A, B)$ such that $\PLD_{(P', Q')}$ is supported on $\cE$ and $(P', Q') \npreceq (P, Q)$.
\end{lemma}

\begin{proof}
	Let $(A, B)$ be the result of the $\eps$-DP binary randomized response, i.e.,
	\begin{align*}
		A(0) = B(1) = \frac{e^\eps}{e^\eps + 1}, \\
		A(1) = B(0) = \frac{1}{e^\eps + 1}.
	\end{align*}
	It is simple to verify that
	\begin{align*}
		h_{(A, B)}(\alpha) =
		\begin{cases}
			1 - \alpha & \text{ if } \alpha \leq e^{-\eps}, \\
			\frac{e^\eps}{e^\eps + 1} - \frac{\alpha}{e^\eps + 1} & \text{ if } e^{-\eps} < \alpha < e^{\eps}, \\
			0 & \text{ if } \alpha \geq \eps.
		\end{cases}
	\end{align*}
	Let $\cA$ be $\{0, \alpha_1, \alpha_2 ,+\infty\}$ where $\alpha_1 = e^{-\eps} - \gamma, \alpha_2 = e^{-\eps} + \gamma$ for any $\gamma < \min\{e^{\eps} - e^{-\eps}, e^{-\eps}\}$. 
	
	Let $h_1: \cA \cup \{+ \infty\} \to [0, 1]$ be defined as
	\begin{align*}
		h_1(0) &= 1, \\
		h_1(\alpha_1) &= h_{(A, B)}(\alpha_1), \\
		h_1(\alpha_2) &= 1 - \alpha_2,\\
		h_1(+\infty) &= 0,
	\end{align*}
	and let $\oh_1$ be its piecewise-linear extension. It is again simple to verify that $\oh_1$ satisfies the conditions in \Cref{lem:hockey-stick-characterization-finite-support} and therefore $\oh_1 = h_{(P_1, Q_1)}$ for some $P_1, Q_1$ such that $\PLD_{(P_1, Q_1)}$ is supported on $\cE$. Furthermore, it can be checked from our definition that $(A, B) \succeq (P_1, Q_1)$.
	
	Similarly, let $h_2: \cA \cup \{+ \infty\} \to [0, 1]$ be defined as
	\begin{align*}
		h_2(0) &= 1, \\
		h_2(\alpha_1) &= 1 - e^{-\eps} + \frac{\gamma}{e^{\eps} + 1}, \\
		h_2(\alpha_2) &= h_{(A, B)}(\alpha_2), \\
		h_2(+\infty) &= 0,
	\end{align*}
	and let $\oh_2$ be its piecewise-linear extension. Again, $\oh_2 = h_{(P_2, Q_2)}$ for some $P_2, Q_2$ such that $\PLD_{(P_2, Q_2)}$ is supported on $\cE$ and $(A, B) \succeq (P_2, Q_2)$.
	
	Now, consider any $(P, Q) \preceq (A, B)$ such that $\PLD_{(P, Q)}$ is supported on $\cE$. We claim that $(P, Q) \nsucceq (\hP_1, \hQ_1)$ or $(P, Q) \nsucceq (P_2, Q_2)$. To prove this, assume for the sake of contradiction that $(P, Q) \nsucceq (P_1, Q_1)$ and $(P, Q) \nsucceq (P_2, Q_2)$. This means that
	\begin{align*}
		h_{(P, Q)}(\alpha_1) &\geq h_1(\alpha_1) = h_{(A, B)}(\alpha_1), \\
		h_{(P, Q)}(\alpha_2) &\geq h_2(\alpha_2)= h_{(A, B)}(\alpha_2). 
	\end{align*}
	From piecewise-linearity of $h_{(P, Q)}$ restricted to $[\alpha_1, \alpha_2]$ (\Cref{lem:hockey-stick-characterization-finite-support}), we then have $h_{(P, Q)}(e^{\eps}) = \frac{1}{2} \left(h_{(P, Q)}(\alpha_1) + h_{(P, Q)}(\alpha_2)\right) > h_{(A, B)}(e^{\eps})$, a contradiction to the assumption that $(P, Q) \preceq (A, B)$.
\end{proof}

\subsection{A Greedy and Convex Hull Construction}

The previous lemma shows that, unfortunately, there is no canonical choice for an optimistic estimate for a given PLD. Due to this, we propose a simple greedy algorithm to construct an optimistic estimate for the PLD of a given pair $(A, B)$. Similar to before, it will be more convenient to deal directly with the hockey-stick curve. Here we would like to construct $f: \cA \to [0, 1]$ such that its piecewise-linear extension $\of$ point-wise lower bounds $h_{(A, B)}$. The distribution $P, Q$ (and $\PLD_{(P, Q)}$) can then be computed using \Cref{alg:discretize}.

Our algorithm will assume that we can compute the derivative of $h$ at any given point $\alpha \in \R_{\geq 0}$ (denoted by $h'(\alpha)$). We remark that, for many widely used mechanisms including Laplace and Gaussian mechanisms, the closed-formed formula for $h'(\alpha)$ can be easily computed. \newline

\noindent
\textbf{First Greedy Attempt.}
Before describing our algorithm, let us describe an approach that does \emph{not} work; this will demonstrate the hurdles we have to overcome. Consider the following simple greedy algorithm: start with $f(\alpha_0) = 0$ and, if we are currently at  $f(\alpha_i)$, then find the largest possible $f(\alpha_{i + 1})$ such that the line $f(\alpha_i), f(\alpha_{i + 1})$ is below $h_{(A, B)}$. (In other words, the line $f(\alpha_i), f(\alpha_{i + 1})$ is tangent to $h_{(A, B)}$.) 

While this is a natural approach, there are two issues with this algorithm:
\begin{itemize}
	\item First and more importantly, it is possible that at some discretization point $f(\alpha_{i + 1})$ becomes negative!  Obviously this invalidates the construction as $f$ will not correspond to a hockey-stick curve.
	\item Secondly, each computation of tangent line requires several (and sequential) computations of the derivative $h'$---rendering the algorithm inefficient---and is also subject
	to possible numerical instability.
\end{itemize}

We remark that if we instead start from right (i.e., $f(\alpha_k)$) and proceed greedily to the left (in decreasing order of $i$), then the first issue will become that $f(\alpha_i)$ can be smaller than $[1 - \alpha_i]_+$, which also makes it an invalid hockey-stick curve due to \Cref{lem:hockey-stick-characterization}(iii). As will be explained below, we will combine these two directions of greedy together with a convex hull algorithm in our revised approach. \newline

\noindent
\textbf{An Additional Assumption.} For our algorithm, we will also need a couple of assumptions. The first one is that 1 belongs to the discretization set:
\begin{assumption} \label{as:one}
	$1 \in \cA$ (or equivalently $0 \in \cE$).
\end{assumption}

For the remainder of this section, we will use $i^* \in [k]$ to denote the index for which $\alpha_{i^*} = 1$ (i.e., $\eps_{i^*} = 0$).

We show below that this assumption is necessary. When it does not hold, then it may simply be impossible to find an optimistic estimate at all:

\begin{lemma}
	There exists a pair $(A, B)$ of distributions such that any $(P, Q) \preceq (A, B)$ satisfies $0 \in \supp(\PLD_{(P, Q)})$. 
\end{lemma}

\begin{proof}
	Let $A = B$. We simply have $h_{(A, B)}(\alpha) = [1 - \alpha]_+$. Due to \Cref{lem:hockey-stick-characterization}(iii), we must have $h_{(P, Q)}(\alpha) = [1 - \alpha]_+$. This simply implies $\PLD_{(P, Q)}(0) = 1$ as desired.
\end{proof}

\iffalse
The second assumption we will need is that $h(+\infty) = 0$:
\begin{assumption} \label{as:two}
	$h(+\infty) = 0$
\end{assumption}
We stress that it is possible to adapt our algorithm to work with the case $h(+\infty) > 0$ but it complicates the algorithm quite a lot and therefore we only operate under this assumption, which already suffices for Gaussian mechanisms and Laplace mechanisms. We also note that, even when $h(+\infty) > 0$, our algorithm as is can also produce a valid pessimistic estimate, as long as $f^{\rightarrow}()$
\fi

\noindent
\textbf{Our Greedy + Convex Hull Algorithm.}
We are now ready to describe our final algorithm. The main idea is to not attempt to create the curve in one left-to-right or right-to-left sweep, but rather to simply generate a ``candidate set'' $F_i$ for each $f(\alpha_i)$. (Such a set will in fact be a singleton for all $i$ except $i = i^*$, for which $|F_i| \leq 2$, but we will refer to $F_i$'s as sets here for simplicity of notation.) We then compute the convex hull of these points $\{(\alpha_i, f_i)\}_{i \in [k - 1], f_i \in F_i}$ and take it (or more precisely its lower curve) as our optimistic estimate. This last step immediately ensures the convexity of our curve, which is required for it to be a valid hockey-stick curve (\Cref{lem:hockey-stick-characterization}(i)).

To construct the candidate set $F_i$, we combine the left-to-right and right-to-left greedy approaches. Specifically, for $i = \{0, \dots, i^* - 1\}$, we draw the tangent line of the true curve $h$ at $h(\alpha_i)$ and let its intersection with the vertical line $\alpha = \alpha_{i + 1}$ be $(\alpha_{i + 1}, f^{\rightarrow}_{i + 1})$; then we add $f^{\rightarrow}_{i + 1}$ into $F_{i + 1}$. Similarly, for $i = \{k - 1, \dots, i^* + 1\}$, we draw the tangent line of the true curve $h$ at $h(\alpha_i)$ and let its intersection with the vertical line $\alpha = \alpha_{i - 1}$ be $(\alpha_{i - 1}, f^{\leftarrow}_{i - 1})$; then we add $f^{\leftarrow}_{i - 1}$ into $F_{i - 1}$. At the very end points $i = 0$ and $i = k - 1$, we also add $1$ and $0$ respectively to $F_i$.

Notice that this algorithm, unlike the previous (failed) greedy approach, only requires a calculation of $h'$ at each $\alpha_i \in \cA \setminus \{1, +\infty\}$, which can be done in parallel. Furthermore, efficient algorithms for convex hull are well known in the literature and can be used directly.

The complete and more precise description of our algorithm is given in \Cref{alg:optimistic}. We also note here that $|F_i| = 1$ for all $i \ne i^*$ and $|F_{i^*}| \leq 2$ (as the point constructed from the left may be different from the point from the right). Nonetheless, we write $F_i$'s as sets for simplicity of notation. An illustration of the algorithm can be found in \Cref{fig:opt-step-by-step}.

\newcommand{\nxt}{\mathrm{next}}
\newcommand{\tmp}{\mathrm{temp}}
\begin{algorithm}
	\caption{Optimistic PLD Construction.} \label{alg:optimistic}
	\begin{algorithmic}
		\Procedure{OptimisticPLD}{$h, \cA$}
		\For{$i = 0, \ldots i^* - 1$}
		\State $f^{\rightarrow}_{i + 1} = h(\alpha_i) + (\alpha_{i + 1} - \alpha_i) \cdot h'(\alpha_i)$
		\EndFor
		\State $f^{\rightarrow}_{0} \gets 1$ \Comment{$h(\alpha_0) = 1$}
		\For{$i = k - 1, \ldots i^* + 1$}
		\State $f^{\leftarrow}_{i - 1} = h(\alpha_i) - (\alpha_i - \alpha_{i - 1}) \cdot h'(\alpha_i)$
		\EndFor
		\State $f^{\leftarrow}_{k - 1} \gets 0$
		\State $H \gets $ ConvexHull($\{(\alpha_i, f^{\rightarrow}_i)\}_{i \in \{0, \dots, i^*\}} \cup \{(\alpha_i, f^{\leftarrow}_i)\}_{i \in \{i^*, \dots, k - 1\}}$)
		\For{$i = 0, \ldots, k - 1$}
		\State $(\alpha_i, f(\alpha_i)) \gets$ lowest intersection point between $H$ and the vertical line $\alpha = \alpha_i$ 
		\EndFor
		\State $f(\alpha_k) \gets 0$ \Comment{$\alpha_k = +\infty$}
		\Return{DiscretizePLD($f(\alpha_0), \dots, f(\alpha_k)$)}
		\EndProcedure
	\end{algorithmic}
\end{algorithm}

Having described our algorithm, we will now proof its correctness, i.e., that it outputs a pair of distributions dominated by the input pair.

\begin{theorem} \label{thm:opt-correctness}
	Let $(P^{\downarrow}, Q^{\downarrow})$ denote the output of $\textsc{\upshape OptimisticPLD}(h, \cA)$ where $h = h_{(A, B)}$. Then, under \Cref{as:one}, we have $(P^{\downarrow}, Q^{\downarrow}) \preceq (A, B)$.
\end{theorem}

To prove \Cref{thm:opt-correctness}, it will be crucial to have the following lower bounds on the candidate points.

\begin{lemma} \label{lem:opt-candidate-points-lower-bound}
	\begin{enumerate}[(i)]
		\item For all $i \in \{0, \dots, i^*\}$, $f^{\rightarrow}(\alpha_i) \geq 1 - \alpha_i$.
		\item For all $i \in \{i^*, \dots, k - 1\}$, $f^{\leftarrow}(\alpha_i) \geq 0$.
		\item For all $i \in \{0, \dots, i^*\}$, $f^{\rightarrow}(\alpha_i) \leq h(\alpha_i)$.
		\item For all $i \in \{i^*, \dots, k - 1\}$, $f^{\leftarrow}(\alpha_i) \leq h(\alpha_i)$.
	\end{enumerate}
\end{lemma}

\begin{proof}
	\begin{enumerate}[(i)]
		\item The statement obviously holds for $i = 0$. Next, consider $i \in [i^*]$. From (ii) and (iii) of \Cref{lem:hockey-stick-characterization}, we have $h'(0) \geq -1$. Furthermore, the convexity of $h$ (\Cref{lem:hockey-stick-characterization}(i)) implies that $h'(\alpha_{k - 1}) \geq h'(0) \geq -1$. Therefore, we have
		\begin{align*}
			f^{\rightarrow}(\alpha_i) &= h(\alpha_{i - 1}) + (\alpha_i - \alpha_{i - 1})\cdot h'(\alpha_{i - 1}) \\
			&\geq h(\alpha_{i - 1}) - (\alpha_i - \alpha_{i - 1}) \\
			&\geq (1 - \alpha_{i - 1}) - (\alpha_i - \alpha_{i - 1}) \\
			&= 1 - \alpha_i,
		\end{align*}
		where the second inequality is due to  \Cref{lem:hockey-stick-characterization}(iii).
		\item The statement obviously holds for $i = k - 1$. Next, consider $i \in \{i^*, \dots, k - 2\}$. Since $h$ is non-increasing (from~\Cref{lem:hockey-stick-characterization}(i)), we have $h'(\alpha_{i + 1}) \leq 0$. Therefore, $f^{\leftarrow}(\alpha_i) \geq h(\alpha_{i + 1}) \geq 0$.
		\item The statement obviously holds for $i = 0$. For $i \in [i^*]$, the convexity of $h$ (\Cref{lem:hockey-stick-characterization}(i)) immediately implies that
		\begin{align*}
			f^{\rightarrow}(\alpha_i) = h(\alpha_{i - 1}) + (\alpha_i - \alpha_{i - 1})\cdot h'(\alpha_{i - 1}) \leq h(\alpha_i).
		\end{align*}
		\item The statement obviously holds for $i = k - 1$. For $i \in \{i^*, \dots, k - 2\}$, the convexity of $h$ (\Cref{lem:hockey-stick-characterization}(i)) immediately implies that
		\begin{align*}
			f^{\rightarrow}(\alpha_i) & = h(\alpha_{i + 1}) + (\alpha_{i + 1} - \alpha_i)\cdot h'(\alpha_{i + 1}) \\
			& \leq h(\alpha_i).
			\qedhere
		\end{align*}
	\end{enumerate}
\end{proof}

We are now ready to prove \Cref{thm:opt-correctness}.

\begin{proof}[Proof of \Cref{thm:opt-correctness}]
	We start by observing that the vertices of the convex hull $H$ consists of the points $(\alpha_{\ell_0}, f^{\rightarrow}(\alpha_{\ell_0})), \dots, (\alpha_{\ell_q}, f^{\rightarrow}(\alpha_{\ell_q})), (\alpha_{\ell_{q + 1}}, f^{\leftarrow}(\alpha_{\ell_{q + 1}})),$ $\dots,  (\alpha_{\ell_m}, f^{\leftarrow}(\alpha_{\ell_m}))$, where $0 = \ell_0 < \cdots < \ell_m = k - 1$ and $\ell_q \leq i^* \leq \ell_{q + 1}$.
	
	First, we have to show that $f(\alpha_0), \dots, f(\alpha_k)$ constitute a valid input to the \textsc{DiscretizePLD} algorithm. Per~\Cref{lem:hockey-stick-characterization-finite-support}, we only need to show that (i) $\of$ is non-increasing, (ii) $\of$ is convex, and (iii) $\of(\alpha) \geq [1 - \alpha]_+$ for all $\alpha \in \R_{\geq 0}$. Notice also that $\of$ is simply the piecewise-linear curve connecting $(\alpha_{\ell_0}, f^{\rightarrow}(\alpha_{\ell_0})), \dots, (\alpha_{\ell_q}, f^{\rightarrow}(\alpha_{\ell_q})), (\alpha_{\ell_{q + 1}}, f^{\leftarrow}(\alpha_{\ell_{q + 1}})),$ $\dots,  (\alpha_{\ell_m}, f^{\leftarrow}(\alpha_{\ell_m}))$.

	To see that (i) holds, observe that the second-rightmost point in the convex hull must be $(\alpha_j, f_j)$ for some $j \in \{0, \dots, k - 2\}$ where $f_j = f^{\rightarrow}_j$ or $f_j = f^{\leftarrow}_j$. In either case, \Cref{lem:opt-candidate-points-lower-bound} implies that $f_j \geq 0 = f^{\leftarrow}(\alpha_{k - 1})$. Since $\of$ is the lower curve of the convex hull $H$ and $(\alpha_j, f_j), (\alpha_{k - 1}, 0)$ is its rightmost segment, we can conclude that $\of$ is non-increasing in the range $[\alpha_0, \alpha_{k - 1}]$.
	Finally, since we simply have $\of(\alpha) = 0$ for all $\alpha > \alpha_{k - 1}$, it is also non-increasing in the range $[\alpha_{k - 1}, +\infty)$, thereby proving (i).
	
	As for (ii), since $\of|_{[\alpha_0, \alpha_{k - 1}]}$ forms the lower boundary of the convex hull $H$, $\of$ is convex in the range $[\alpha_0, \alpha_{k - 1}]$. Again, since we simply have $\of(\alpha) = 0$ for all $\alpha > \alpha_{k - 1}$, we can conclude that it is convex for the entire range $[0, +\infty)$.
	
	Finally, for (iii),~\Cref{lem:opt-candidate-points-lower-bound} states that all the points $\{(\alpha_i, f^{\rightarrow}_i)\}_{i \in \{0, \dots, i^*\}} \cup \{(\alpha_i, f^{\leftarrow}_i)\}_{i \in \{i^*, \dots, k - 1\}}$ is above the curve $\alpha \mapsto [1 - \alpha]_+$. Since $\of$ is in the convex hull $H$, we can conclude that $\of$ also lies above this curve, as desired.
	
	Now that we have proved that $f(\alpha_0), \dots, f(\alpha_k)$ is a valid input to the \textsc{DiscretizePLD} algorithm (and therefore the output $(P^{\uparrow}, Q^{\uparrow})$ is a pair of valid probability distributions), we will next show that $(P^{\uparrow}, Q^{\uparrow}) \preceq (A, B)$. This is equivalent to showing that $\of(\alpha) \leq h(\alpha)$ for all $\alpha \geq \R_{\geq 0} \cup \{+ \infty\}$.
	To prove this, we consider three cases based on the value of $\alpha$. For brevity, we say that a curve $C_1$ is \emph{below} a curve $C_2$ when they share the same domain $\Omega$ and $C_1(o) \leq C_2(o)$ for all $o \in \Omega$.
	\begin{itemize}
		\item Case I: $\alpha \geq \alpha_{k - 1}$. In this case, $\of(\alpha) = 0 \leq h(\alpha)$.
		\item Case II: $\alpha \in [1, \alpha_{k - 1})$. Suppose that $\alpha \in [\alpha_i, \alpha_{i + 1})$. Let $L_1$ denote the line segment from $(\alpha_i, f^{\leftarrow}(\alpha_i))$ to $(\alpha_{i + 1}, f^{\leftarrow}(\alpha_{i + 1}))$, and $L_2$ denote the line segment from $(\alpha_i, f^{\leftarrow}(\alpha_i))$ to $(\alpha_{i + 1}, h(\alpha_{i + 1}))$. From \Cref{lem:opt-candidate-points-lower-bound}(iv), $L_1$ is below $L_2$. Furthermore, since $\of|_{[\alpha_i, \alpha_{i + 1})}$ is a lower boundary of the convex hull $H$ containing $L_1$, it must also be below $L_1$. Therefore, we have
		\begin{align*}
			\of(\alpha) \leq L_2(\alpha) &\leq L_1(\alpha) \\
			&= h(\alpha_{i + 1}) - (\alpha_{i + 1} - \alpha) h'(\alpha_{i + 1}) \\
			&\leq h(\alpha),
		\end{align*}
		where the last inequality follows from convexity of $h$.
		\item Case III: $\alpha \in [0, 1)$. Suppose that $\alpha \in [\alpha_i, \alpha_{i + 1})$. Similarly to the previous case, let $L_3$ denote the line segment from $(\alpha_i, f^{\rightarrow}(\alpha_i))$ to $(\alpha_{i + 1}, f^{\rightarrow}(\alpha_{i + 1}))$, and $L_4$ denote the line segment from $(\alpha_i, h(\alpha_i))$ to $(\alpha_{i + 1}, f^{\rightarrow}(\alpha_{i + 1}))$. From \Cref{lem:opt-candidate-points-lower-bound}(iii), $L_3$ is below $L_4$. Furthermore, since $\of|_{[\alpha_i, \alpha_{i + 1})}$ is a lower boundary of the convex hull $H$ containing $L_3$, it must also be below $L_3$. Therefore, we have
		\begin{align*}
			\of(\alpha) \leq L_3(\alpha) &\leq L_4(\alpha) \\
			&= h(\alpha_i) + (\alpha_{i + 1} - \alpha) h'(\alpha_i) \\
			&\leq h(\alpha),
		\end{align*}
		where the last inequality follows from convexity of $h$.
	\end{itemize}
	As a result, we can conclude that $(P^{\uparrow}, Q^{\uparrow}) \preceq (A, B)$, completing our proof.
\end{proof}

\subsection{Comparison to Privacy Loss Buckets}
\label{sec:pb-optimistic}

Similar to~\Cref{sec:pb-pessimistic}, PB~\cite{meiser2018tight} can also be applied for optimistic-estimate: let $\tPLD^{\downarrow}_{(A,B)}$ be the probability distribution where
\begin{align*}
	\tPLD^{\downarrow}_{(A,B)}(\eps_{i - 1}) = \PLD_{(A,B)}([\eps_{i - 1}, \eps_i)),
\end{align*}
for all $i \in [k]$. That is, $\tPLD^{\downarrow}_{(A,B)}$ is a probability distribution on $\cE$ that is stochastically dominated by $\PLD_{(A,B)}$; furthermore, $\tPLD^{\downarrow}_{(A,B)}$ is the greatest such distribution under stochastic dominant (partial) ordering. It is important to note that, unlike the pessimistic-PB estimate, the optimistic-PB estimate $\tPLD^{\downarrow}_{(A,B)}$ is \emph{not} necessarily a valid PLD for some pair of distributions. This can easily be seen by, e.g., taking a PLD of any $\eps$-DP mechanism for finite $\eps$ and let $\cE = \{-\infty, +\infty\}$; the optimistic-PB estimate puts all of its mass at $0$, which is clearly not a valid PLD.

We present illustrations of our optimistic PLD and optimistic-PB estimates in \Cref{fig:opt-simple}. Recall that in the pessimistic case, we can show that the pessimistic-PB estimate is no better than our approach (\Cref{lem:pess-better-than-pb}). Although we observe similar behaviours in the optimistic case  in simple examples (e.g. \Cref{fig:opt-simple}) and also in our experiments in \Cref{sec:exp}, this unfortunately does not hold in general. Indeed, if $\PLD_{(A, B)}$ has a non-zero mass at $+\infty$ (or equivalently $h(+\infty) \ne 0$), then the optimistic-PB estimate still keeps this mass while our does not. The latter is because we set $f(+\infty) = 0$ in \Cref{alg:optimistic}.  Note here that we cannot set $f(+\infty) = h(+\infty)$ here because the monotonicity may not hold anymore; it is possible that $f^{\rightarrow}(\alpha_i) < h(+\infty)$ for some $i \in [i^*]$. Such examples highlight the challenge in finding a good optimistic estimate (especially in light of the non-existence of the best one, i.e., \Cref{lem:opt-not-unique}), and we provide further discussion regarding this in \Cref{sec:open}.

\begin{figure*}
	\centering
	\begin{subfigure}[b]{0.48\textwidth}
		\centering
		\includegraphics[width=\textwidth]{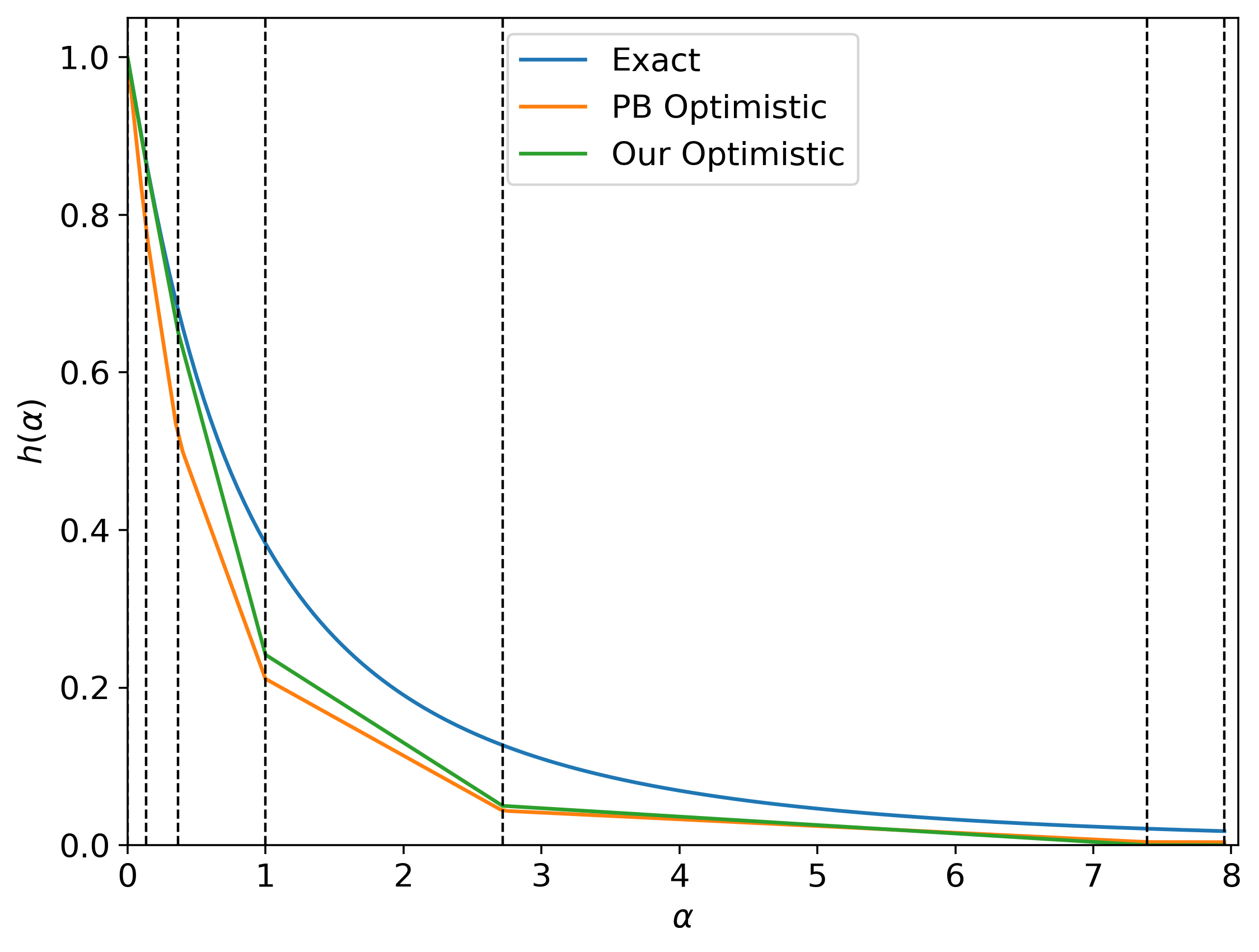}
		\caption{Gaussian Mechanism}
		\label{fig:opt-gaussian}
	\end{subfigure}
	\hfill
	\begin{subfigure}[b]{0.48\textwidth}
		\centering
		\includegraphics[width=\textwidth]{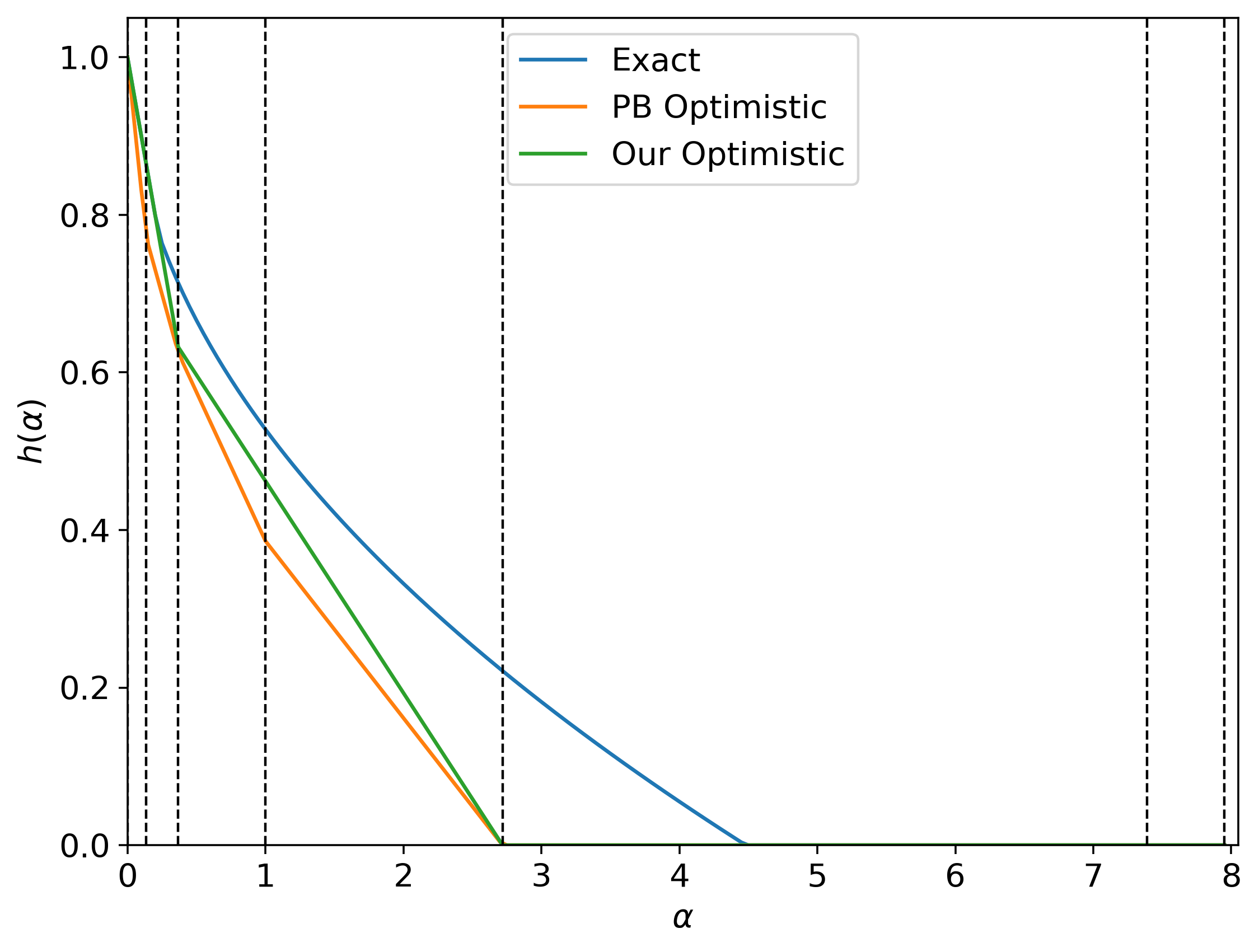}
		\caption{Laplace Mechanism}
		\label{fig:opt-laplace}
	\end{subfigure}
	\caption{Illustrations of the hockey-stick curves of the Gaussian and Laplace mechanisms , and their optimistic estimates from our approach  and the Privacy Bucket (PB) approach of~\cite{meiser2018tight}. The setting of parameters and labels are similar to \Cref{fig:pess-simple}.}
	\label{fig:opt-simple}
\end{figure*}

\begin{figure*}
	\centering
	\begin{subfigure}[b]{0.32\textwidth}
		\centering
		\includegraphics[width=\textwidth]{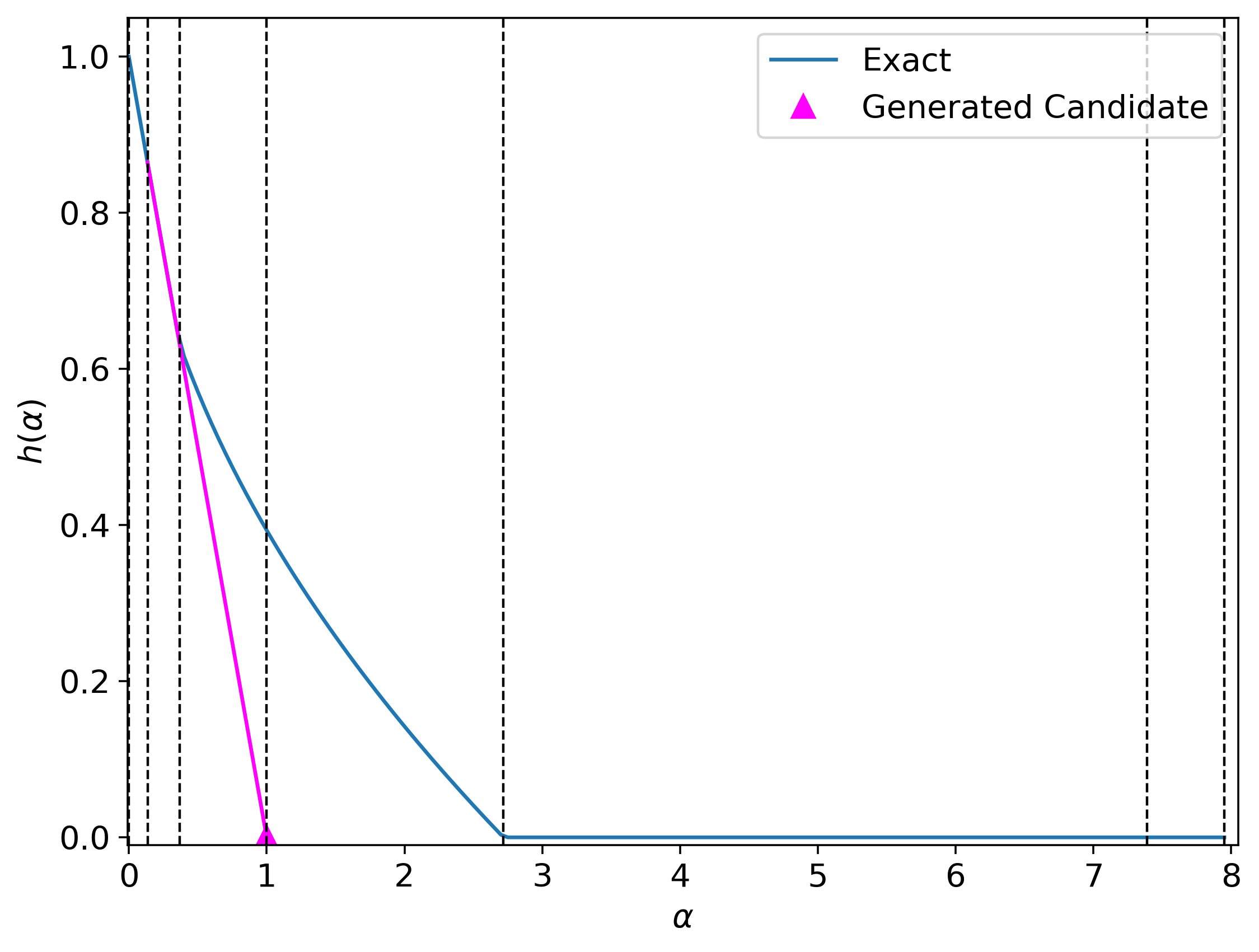}
		\caption{One Step of Candidate Generation for $i < i^*$}
		\label{fig:opt-candidate-forward}
	\end{subfigure}
	\hfill
	\begin{subfigure}[b]{0.32\textwidth}
		\centering
		\includegraphics[width=\textwidth]{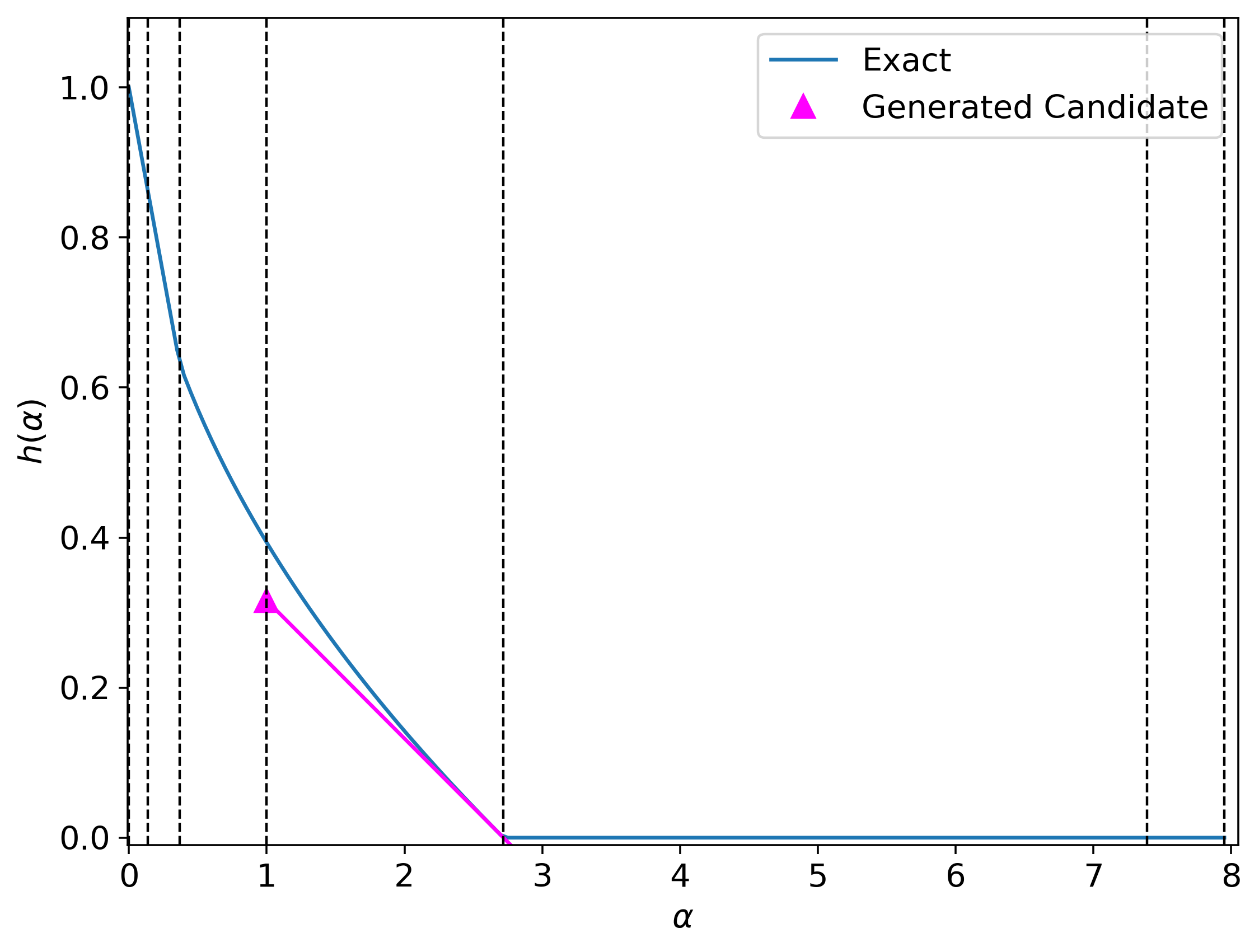}
		\caption{One Step of Candidate Generation for $i > i^*$}
		\label{fig:opt-candidate-backward}
	\end{subfigure}
	\hfill
	\begin{subfigure}[b]{0.32\textwidth}
		\centering
		\includegraphics[width=\textwidth]{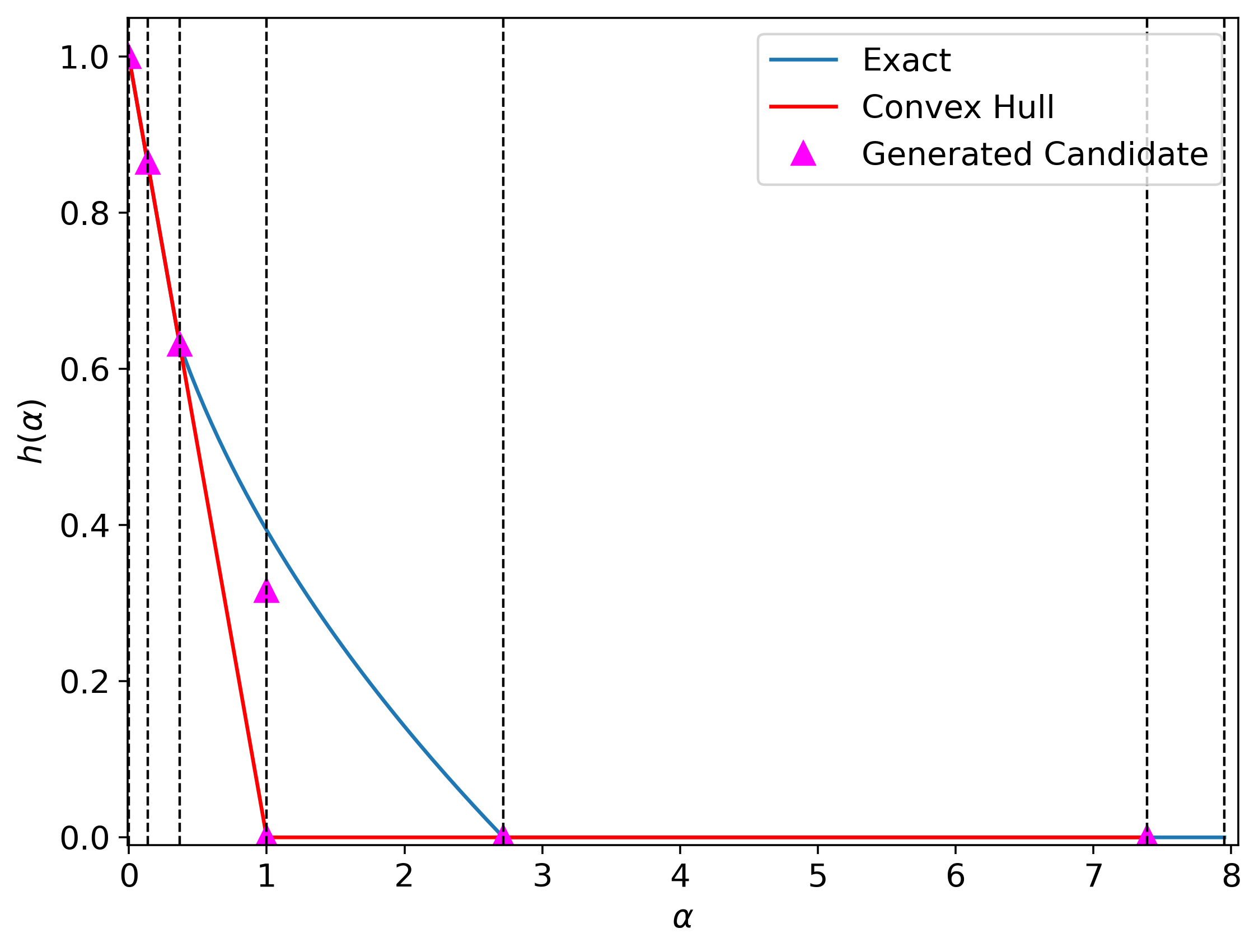}
		\caption{All Candidate Points and Its Convex Hull}
		\label{fig:opt-convex-hull}
	\end{subfigure}
	\caption{Illustrations of our optimistic PLD construction algorithm (\Cref{alg:optimistic}) for Laplace mechanism with noise multiplier 1. \Cref{fig:opt-candidate-forward} demonstrates one step of how the candidate points are generated when $i < i^*$. Specifically, a line tangent to the hockey-stick curve is drawn at each point $(\alpha_i, h(\alpha_i))$; the intersections with the vertical line at $\alpha_{i+1}$ give the candidate points $(\alpha_{i+1}, f^{\rightarrow}_{i+1})$. Similarly, \Cref{fig:opt-candidate-backward} shows such a step for $i > i^*$; in this case, the same line is drawn and its intersection with the vertical line at $\alpha_{i-1}$ give the candidate points $(\alpha_{i-1}, f^{\leftarrow}_{i-1})$. \Cref{fig:opt-convex-hull} shows all the candidate points generated together with its convex hull, which we use as our optimistic estimate.}
	\label{fig:opt-step-by-step}
\end{figure*}

\section{Evaluation}\label{sec:exp}

We compare our algorithm with the Privacy Buckets algorithm~\cite{meiser2018tight} as implemented in the Google DP library\footnote{Even though there are several other papers~\cite{koskela2020computing,koskela2021tight,koskela2021computing} that build on PB, all of them still use the same PB-based approximation, with the differences being how the truncation is computed for FFT. We use the implementation in the Google DP library \href{https://github.com/google/differential-privacy/tree/main/python}{github.com/google/differential-privacy/tree/main/python}}, and the algorithm of Gopi et al.~\cite{gopi2021numerical} implemented in Microsoft PRV Accountant.\footnote{Implementation from \href{https://github.com/microsoft/prv_accountant}{github.com/microsoft/prv\_accountant}}
Gopi et al.'s algorithm does not fit into the pessimistic/optimistic framework as described in~\Cref{sec:accounting-framework}. Instead, their algorithm uses an approximation of PLD that is neither optimistic nor pessimistic and uses a concentration bound to derive pessimistic and optimistic estimates. Indeed, their approximate distribution maintains the same expectation as the true PLD, which is the main ingredient in their improvement over previous work.

\begin{figure}[t]
	\centering
	\includegraphics[width=0.45\textwidth]{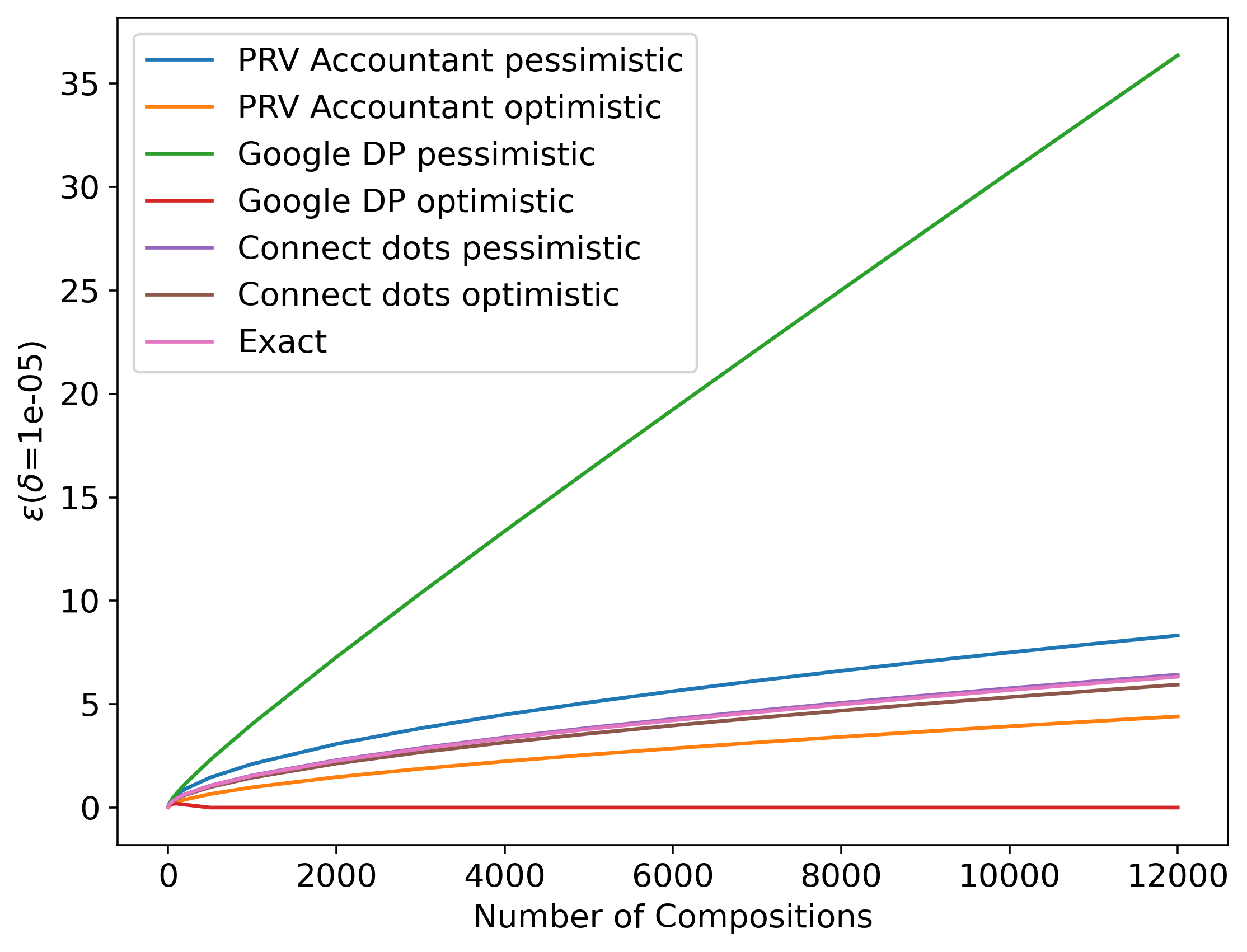}
	\caption{$\eps$ vs Number of Compositions of the Gaussian mechanism with noise scale $80$. All methods evaluated with the same discretization interval of $0.005$.}
	\label{fig:gaussian-same-disc-epsilon}
\end{figure}

As a first cut, we evaluate pessimistic and optimistic estimates on the privacy parameter $\eps$, for a fixed value of $\delta=10^{-5}$, for varying number of compositions of the Gaussian mechanism with noise scale $80$, while comparing our approach to the two other implementations mentioned above. We use the same discretization interval to evaluate each algorithm. The reason for choosing the Gaussian mechanism is that the exact value of $\eps$ can be computed explicitly. We find that the estimates given by our approach are the tightest.

\begin{remark}
For any specified discretization interval, each algorithm has a different choice of how many discretization points are included in $\calE$. Our implementation uses the same set of discretization points as used by the Google DP implementation. The number of discretization points increases with the number of self compositions (we use the Google DP implementation to perform self composition\footnote{We found that the Google DP implementation has a significantly worse running time when computing optimistic estimates, due to lack of truncation. We modify the self-composition method in the Google DP library to incorporate truncation when computing optimistic estimates, and evaluate both ours and Google DP implementation with this minor modification. These do not change the estimates significantly, but drastically reduce the running time.}). On the other hand, Microsoft PRV Accountant chooses a number of discretization points, depending on the number of compositions desired, and this number does not change after self composition. In all the evaluation experiments mentioned in this paper, we find that the number of discretization points in our approach are lower than the number of discretization points in the PRV Accountant (even after composition). 
\end{remark}

Our main evaluation involves computing pessimistic and optimistic estimates on the privacy parameter $\eps$, for a fixed value of $\delta$, for varying number of compositions of the Poisson sub-sampled Gaussian mechanism and comparing our approach to each of the two other implementations.  Note that this particular mechanism is quite popular in that it captures the privacy analysis of DP-SGD where the number of compositions is equal to the number of iterations of the training algorithm, and the subsampling rate is equal to the fraction of the batch size divided by the total number of training examples~\cite{abadi2016deep}. In particular, we consider the Gaussian mechanism with noise scale $1$, Poisson-subsampled with probability $0.01$. We compare against each competing algorithm twice, once where both algorithms use the same discretization interval, and once where our approach uses a larger discretization interval than the competing algorithm. We additionally plot the running time required for this computation for each number of compositions; we ran the evaluation for each number of compositions $20 \times$ and plot the mean running time along with a shaded region indicating 25th--75th percentiles of running time.

\paragraph*{Comparison with Google DP.} The comparison with Google DP is presented in \Cref{fig:subsampgauss-pb}.
Figures \ref{fig:subsampgauss-pb-same-disc-epsilon} and \ref{fig:subsampgauss-pb-same-disc-runtime} compare the $\eps$'s and runtimes for both methods using the same discretization interval, and finds that our method gives a significantly tighter estimate for a mildly larger running time.
Figures \ref{fig:subsampgauss-pb-diff-disc-epsilon} and \ref{fig:subsampgauss-pb-diff-disc-runtime} compare the $\eps$'s and runtimes for both methods with different discretization intervals, and using a discretization interval that is $66.66 \times$ larger, our method gives comparable estimates, with a drastic speed-up ($\sim 300 \times$).

\begin{figure*}
	\centering
	\begin{subfigure}[b]{0.48\textwidth}
		\centering
		\includegraphics[height=0.7\textwidth]{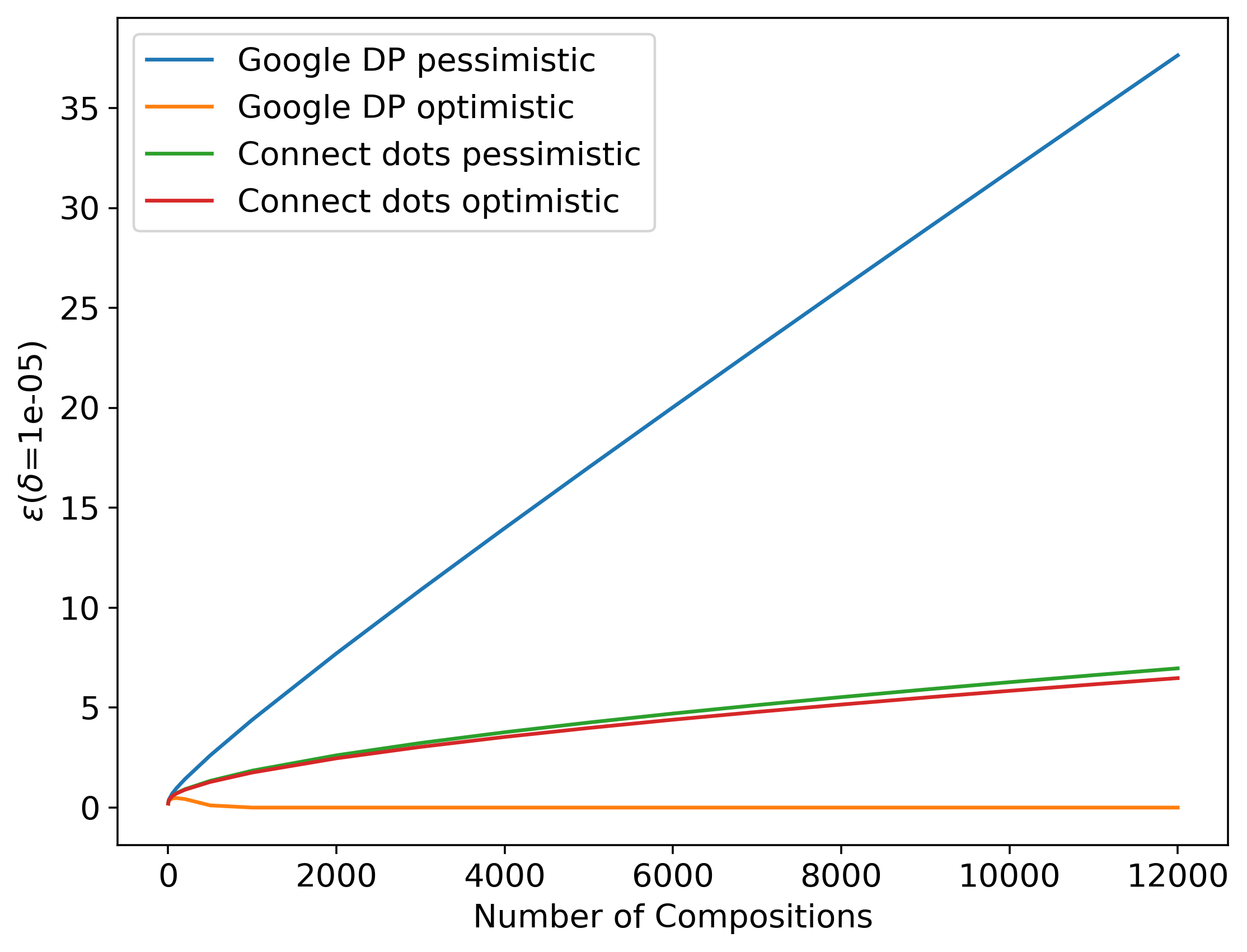}
		\caption{$\eps$ comparison with discretization interval 0.005 for both this work and Google DP}
		\label{fig:subsampgauss-pb-same-disc-epsilon}
	\end{subfigure}
	\hfill
	\begin{subfigure}[b]{0.48\textwidth}
		\centering
		\includegraphics[height=0.7\textwidth]{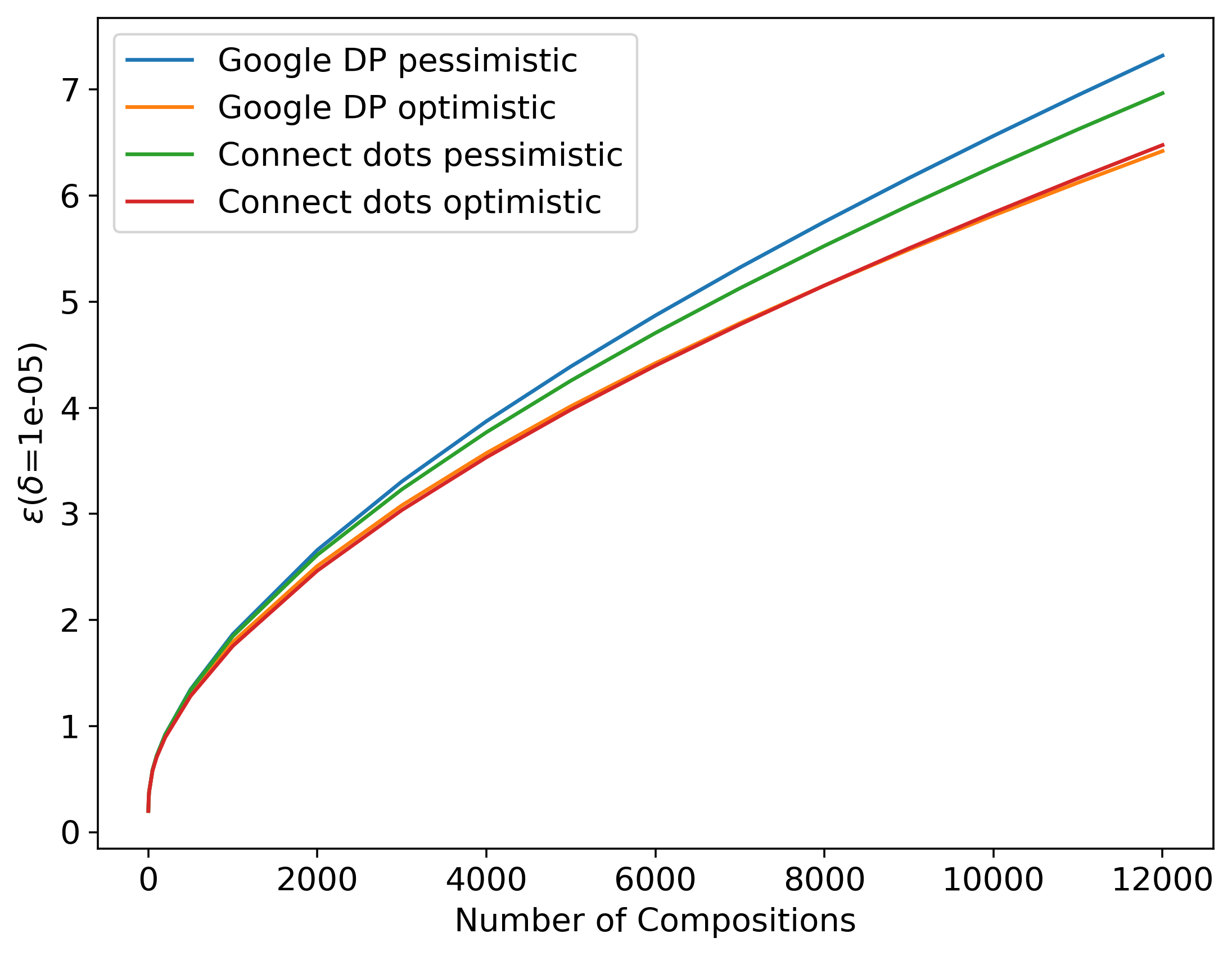}
		\caption{$\eps$ comparison with discretization interval 0.005 for this work, 0.000075 for Google DP}
		\label{fig:subsampgauss-pb-diff-disc-epsilon}
	\end{subfigure}
	\newline
	\begin{subfigure}[b]{0.48\textwidth}
		\centering
		\includegraphics[height=0.7\textwidth]{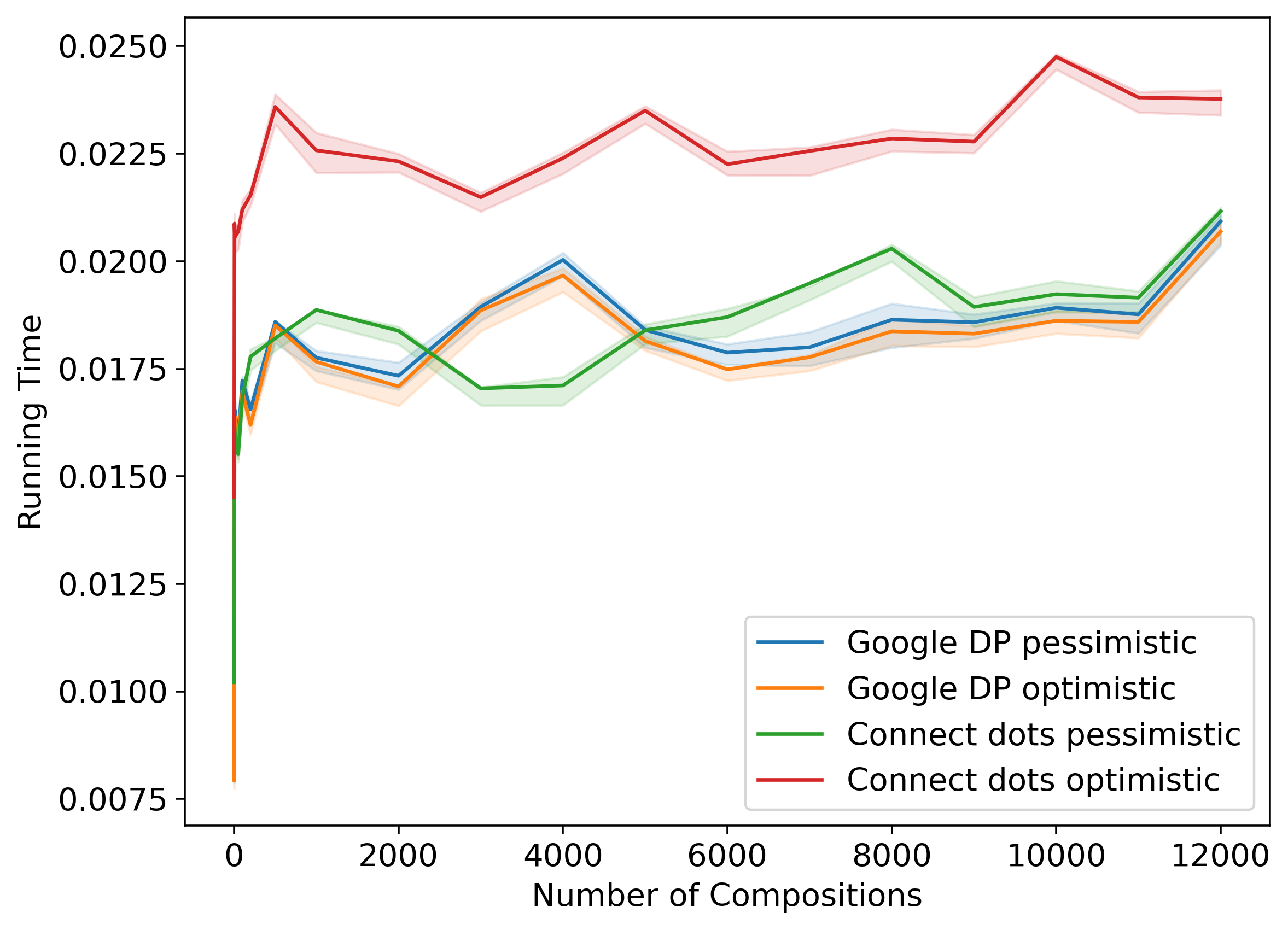}
		\caption{Running time comparison with discretization interval 0.005 for both this work and Google DP (shaded region indicates 25-75 percentile range across 20 independent runs)}
		\label{fig:subsampgauss-pb-same-disc-runtime}
	\end{subfigure}
	\hfill
	\begin{subfigure}[b]{0.48\textwidth}
		\centering
		\includegraphics[height=0.7\textwidth]{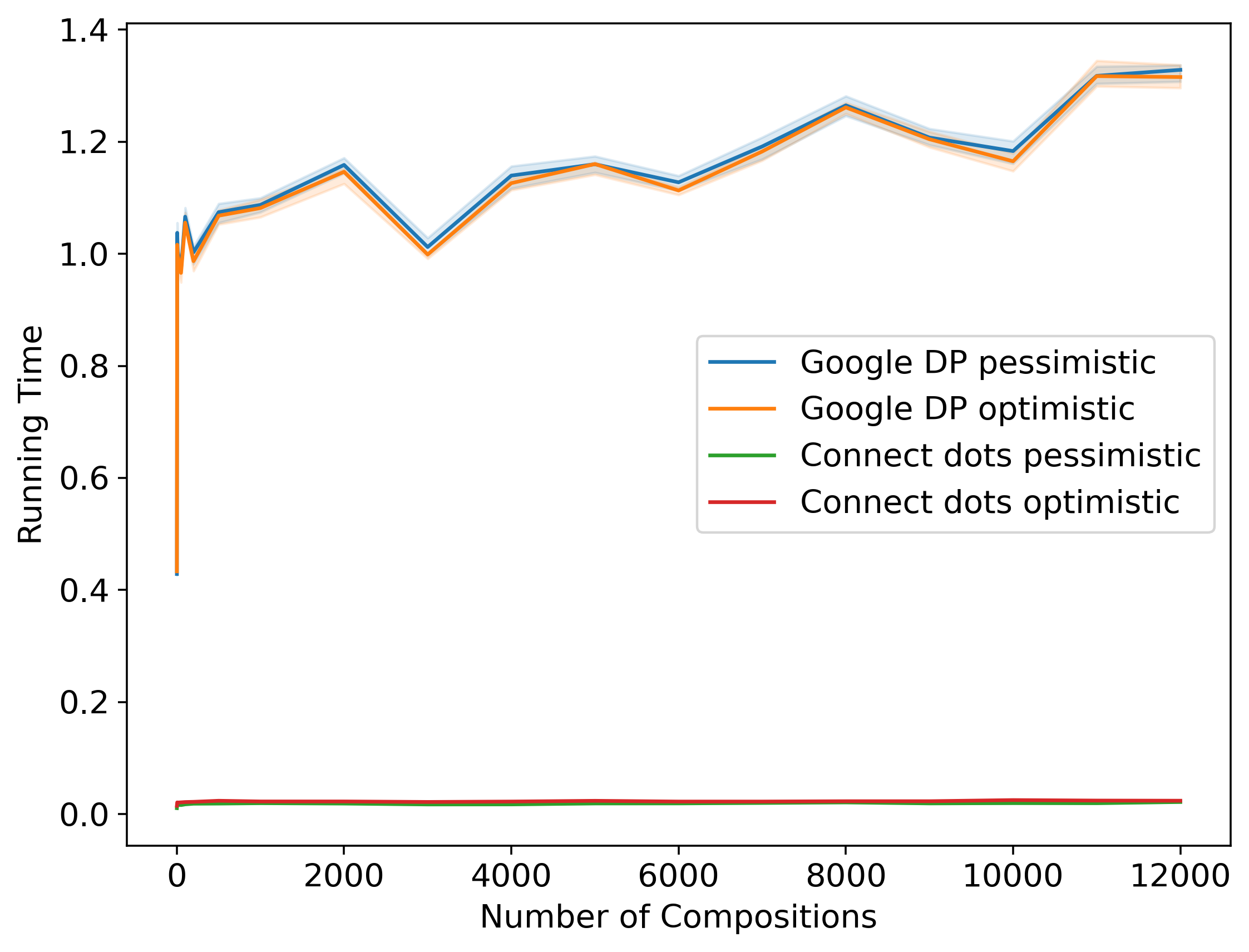}
		\caption{Running time comparison with discretization interval 0.005 for this work, 0.000075 for Google DP (shaded region indicates 25-75 percentile range across 20 independent runs)}
		\label{fig:subsampgauss-pb-diff-disc-runtime}
	\end{subfigure}
	\caption{Computation of pessimistic/optimistic estimates of the privacy parameter $\eps$ (for fixed parameter $\delta = 10^{-5}$) for self-composition of the Gaussian mechanism with noise scale $1$, Poisson-subsampled with probability $0.01$, using the Google DP implementation of the PB approach~\cite{meiser2018tight} vs. our approach. %
	Figures \ref{fig:subsampgauss-pb-same-disc-epsilon} and \ref{fig:subsampgauss-pb-same-disc-runtime} compare the $\eps$'s and runtimes for both methods using the same discretization interval. 
	Figures \ref{fig:subsampgauss-pb-diff-disc-epsilon} and \ref{fig:subsampgauss-pb-diff-disc-runtime} compare the $\eps$'s and runtimes for both methods with different discretization intervals.
	}
	\label{fig:subsampgauss-pb}
\end{figure*}

\paragraph*{Comparison with Microsoft PRV Accountant.} The comparison with Microsoft PRV Accountant is presented in \Cref{fig:subsampgauss-gopi}.
Figures \ref{fig:subsampgauss-gopi-same-disc-epsilon} and \ref{fig:subsampgauss-gopi-same-disc-runtime} compare the $\eps$'s and runtimes for both methods using the same discretization interval, and finds that our method gives a significantly tighter estimate with already shorter running time.
Figures \ref{fig:subsampgauss-gopi-diff-disc-epsilon} and \ref{fig:subsampgauss-gopi-diff-disc-runtime} compare the $\eps$'s and runtimes for both methods with different discretization intervals, and using a discretization interval that is $6.66 \times$ larger, our method gives comparable estimates, with an even larger speed-up.

\begin{figure*}
	\centering
	\begin{subfigure}[b]{0.48\textwidth}
		\centering
		\includegraphics[height=0.7\textwidth]{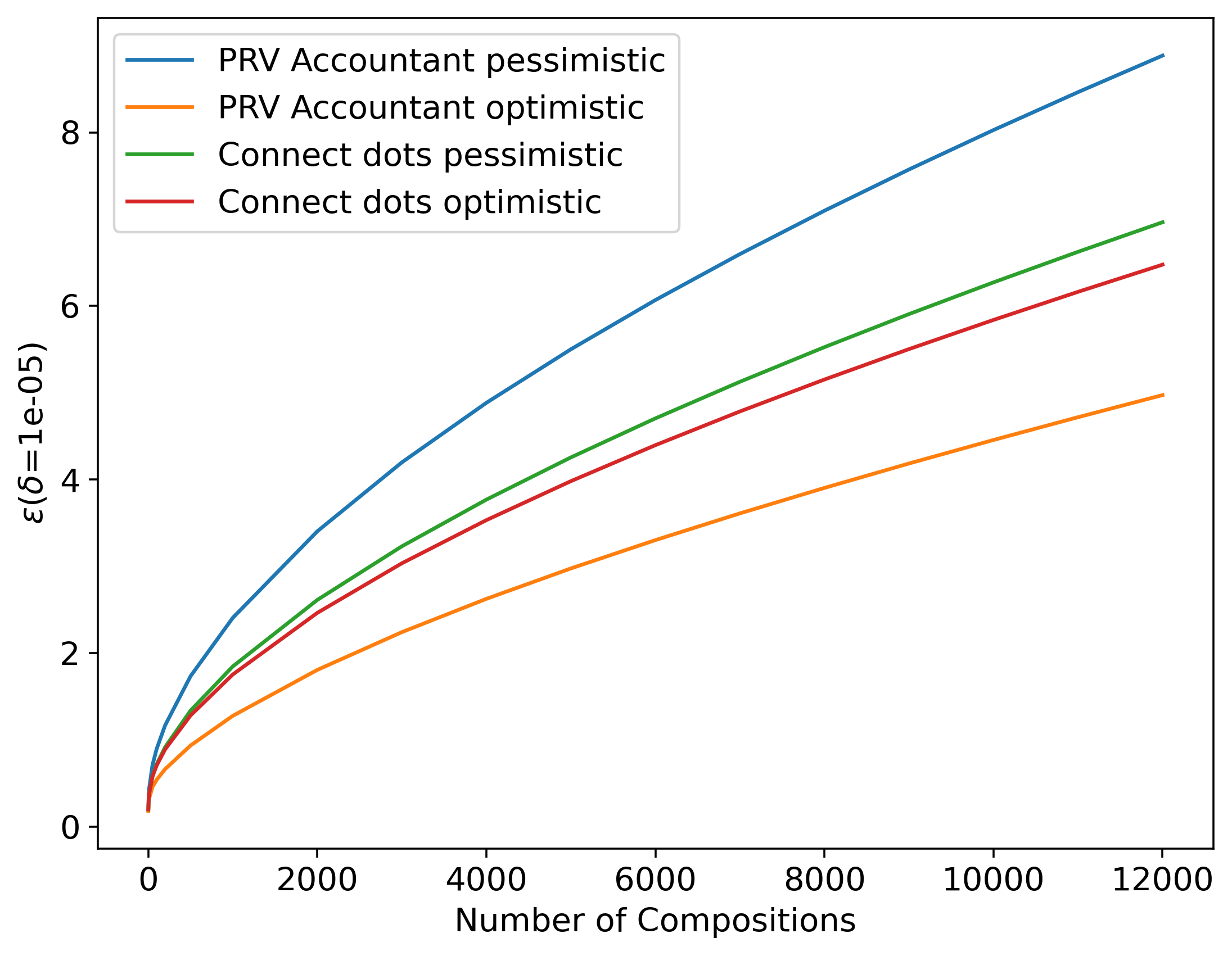}
		\caption{$\eps$ comparison with discretization interval 0.005 for both this work and PRV Accountant}
		\label{fig:subsampgauss-gopi-same-disc-epsilon}
	\end{subfigure}
	\hfill
	\begin{subfigure}[b]{0.48\textwidth}
		\centering
		\includegraphics[height=0.7\textwidth]{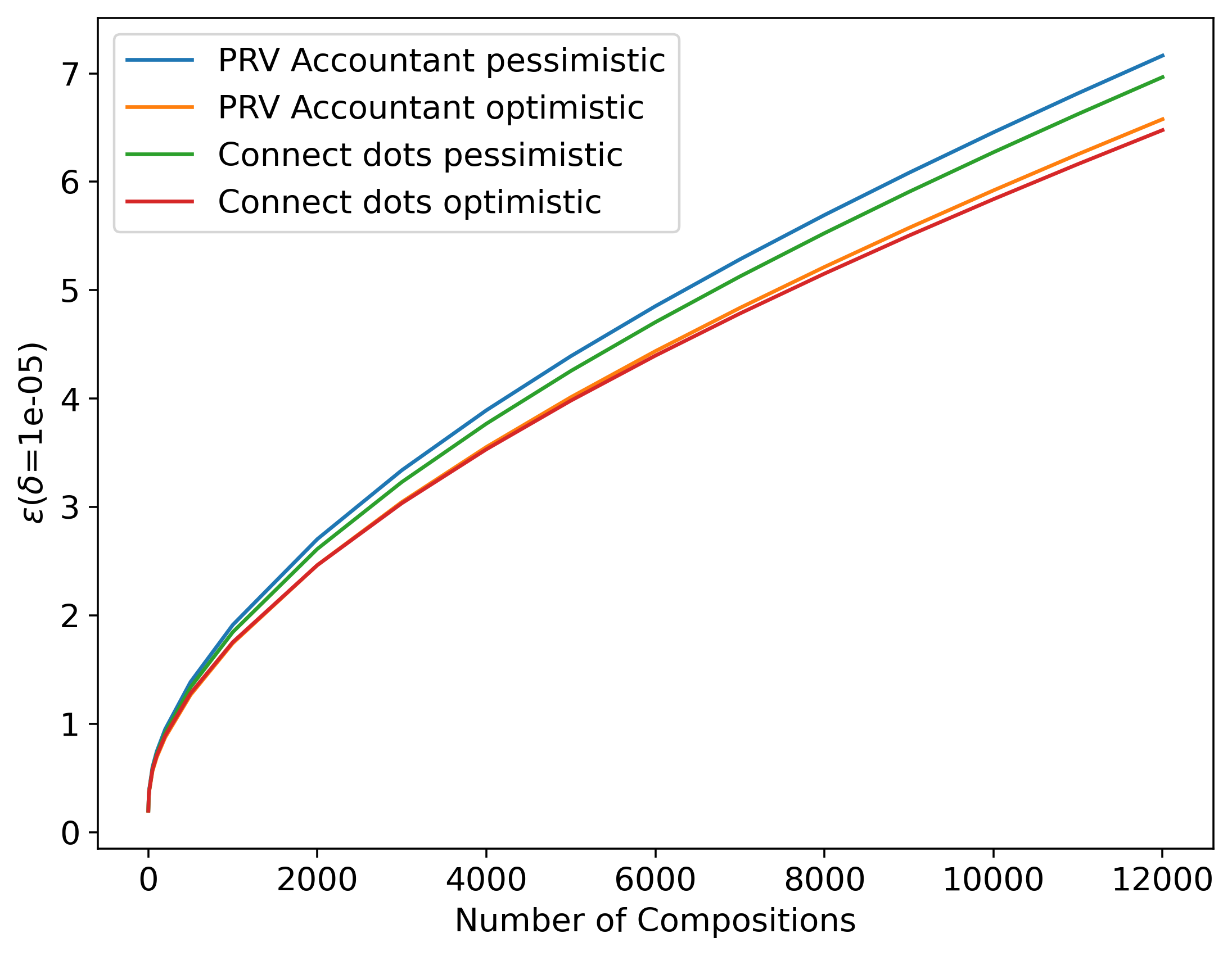}
		\caption{$\eps$ comparison with discretization interval 0.005 for this work, 0.00075 for PRV Accountant (shaded region indicates 25-75 percentile range across 20 independent runs)}
		\label{fig:subsampgauss-gopi-diff-disc-epsilon}
	\end{subfigure}
	\newline
	\begin{subfigure}[b]{0.48\textwidth}
		\centering
		\includegraphics[height=0.7\textwidth]{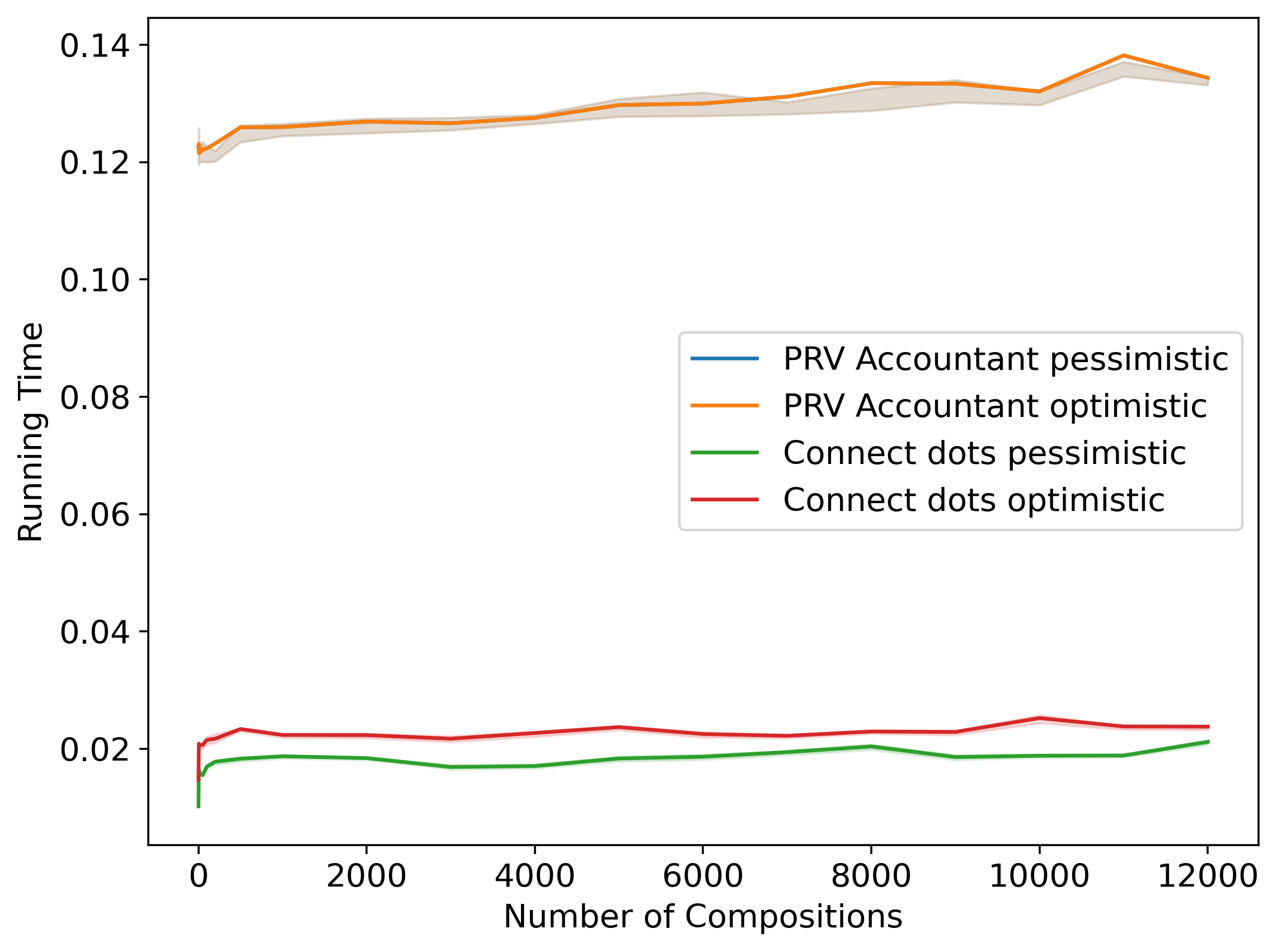}
		\caption{Running time comparison with discretization interval 0.005 for both this work and PRV Accountant (shaded region indicates 25-75 percentile range across 20 independent runs)}
		\label{fig:subsampgauss-gopi-same-disc-runtime}
	\end{subfigure}
	\hfill
	\begin{subfigure}[b]{0.48\textwidth}
		\centering
		\includegraphics[height=0.7\textwidth]{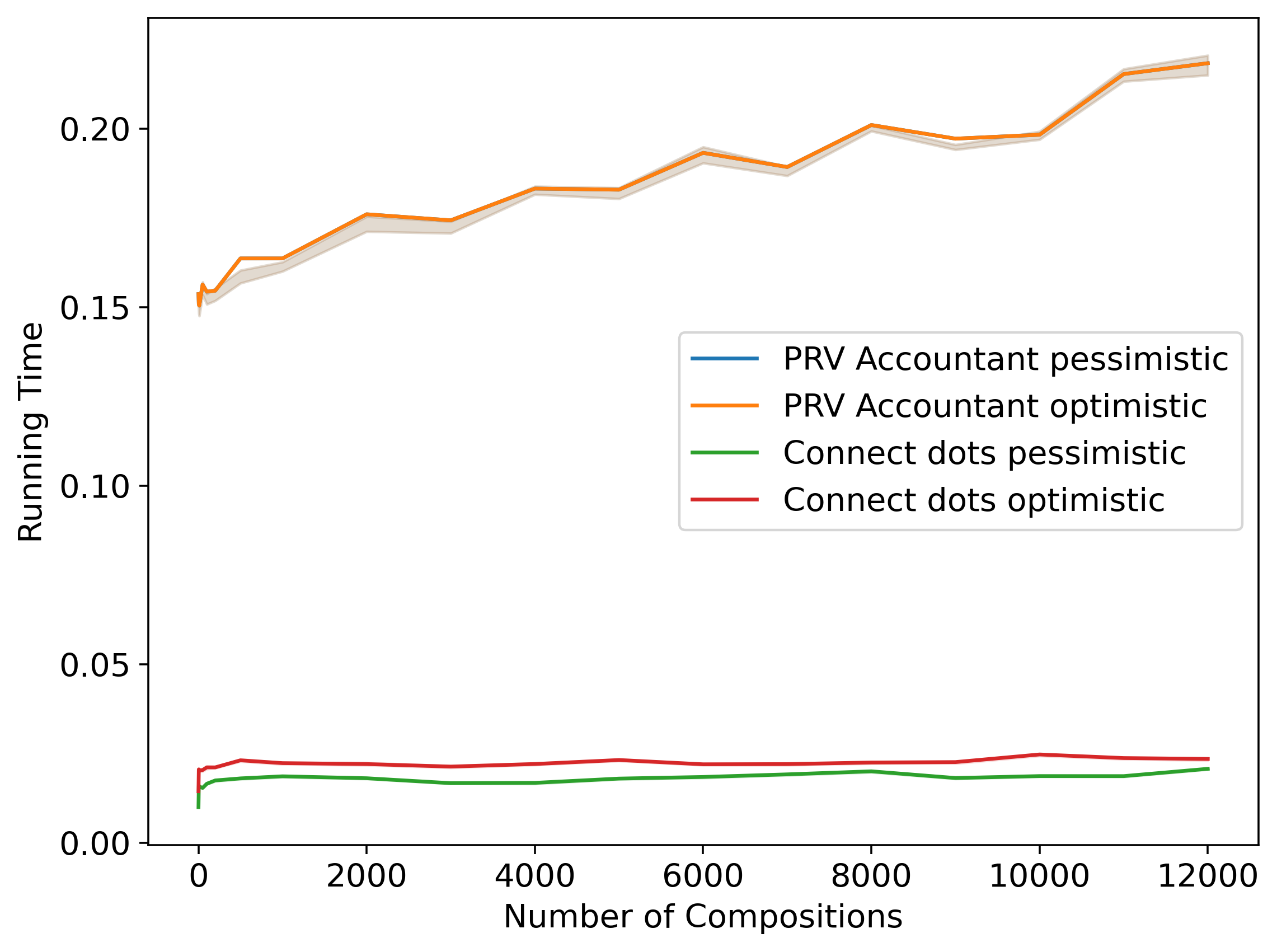}
		\caption{Running time comparison with discretization interval 0.005 for this work, 0.00075 for PRV Accountant}
		\label{fig:subsampgauss-gopi-diff-disc-runtime}
	\end{subfigure}
	\caption{Computation of pessimistic/optimistic estimates of the privacy parameter $\eps$ (for fixed parameter $\delta = 10^{-5}$) for self-composition of the Gaussian mechanism with noise scale $1$, Poisson-subsampled with probability $0.01$, using the Microsoft PRV Accountant~\cite{gopi2021numerical} vs. our approach. We note that the pessimistic and optimistic curves of the Microsoft PRV Accountant~\cite{gopi2021numerical} are identical in \Cref{fig:subsampgauss-gopi-same-disc-runtime} and~\Cref{fig:subsampgauss-gopi-diff-disc-runtime}. %
	}
	\label{fig:subsampgauss-gopi}
\end{figure*}

In \Cref{apx:poisson-subsampled-laplace}, we perform a similar evaluation for the Poisson-subsampled Laplace mechanism.

\section{Discussion \& Open Problems}
\label{sec:open}

In this work, we have proposed a novel approach to pessimistic and optimistic estimates for PLDs, which outperforms previous approaches under similar discretization intervals, and allows for a more compact representation of PLDs while retaining similar error guarantees. There are still several interesting future directions that one could consider.

As we have proved in \Cref{lem:opt-not-unique}, there is no unique ``best'' way to pick an optimistic estimate, and we proposed a greedy algorithm (\Cref{alg:optimistic}) for this task. However, it is difficult to determine how good this greedy algorithm is in general. Instead, it might also be interesting to find $(P^{\downarrow}, Q^{\downarrow}) \preceq (A, B)$ that minimizes a certain objective involving $h_{(P^{\downarrow}, Q^{\downarrow})}$ and $h_{(A, B)}$. For example, one could consider the area between the two curves, or the Fr\'{e}chet distance between them. An intriguing direction here is to determine (i) which objective captures the notion of ``good approximation'' better in terms of composition, and (ii) for a given objective, whether there is an efficient algorithm to compute such $(P^{\downarrow}, Q^{\downarrow})$. We remark that for some objectives, such as the area between the two curves, it is possible to discretize the candidate values for each $f(\alpha_i)$ and use dynamic programming in an increasing order of $i$ (with the state being $f(\alpha_{i - 1}), f(\alpha_i)$). Even for these objectives, it remains interesting to determine whether such discretization is necessary and whether more efficient algorithms exist.

Also related is the question of how to theoretically explain our experimental findings (\Cref{sec:exp}). Although we see significant numerical improvements, it is intriguing to understand theoretically where these improvements come from and which properties of PLDs govern how big such improvements are. More broadly, given an estimate of a PLD, how can we quantify how ``good'' it is? Previous work (e.g.,~\cite{koskela2020computing,gopi2021numerical}) has obtained certain theoretical bounds on the errors; it would be interesting to investigate whether these bounds can help answer the aforementioned question.

Furthermore, the entire line of work on PLD-based accounting~\cite{meiser2018tight,koskela2020computing,koskela2021tight,koskela2021computing}, including this paper, has so far considered only \emph{non-interactive} compositions, meaning that the mechanisms that are run in subsequent steps cannot be changed based on the outputs from the previous steps. This is not a coincidence: interactive composition is highly non-trivial and in fact it is known that the advanced composition theorem (which is even a more specific form of PLD) does not hold in this regime~\cite{RogersVRU16}. Several solutions have been proposed here, such as modified formulae for advanced compositions~\cite{RogersVRU16,whitehouseimproved} and Renyi DP~\cite{renyifilter1,renyifilter2}. However, as discussed earlier, these methods may be loose even in the non-interactive setting, which is the original motivation for PLD-based accounting. Therefore, it would be interesting to understand whether there is a tighter method similar to PLDs that also works in the interactive setting.

\subsection*{Acknowledgments}

This research received no specific grant from any funding agency in the public, commercial, or not-for-profit sectors.

\newpage
\bibliographystyle{plain}
\bibliography{main.bbl}

\newpage

\appendix

\section{Proofs}\label{apx:proofs}

\lemPldHockeyStick*
\begin{proof}
We have from the definition of hockey-stick divergence that
\begin{align*}
D_{e^\eps}(P || Q) &~=~ \sum_{\omega}\ [P(\omega) - e^\eps \cdot Q(\omega)]_{+}\\
&~=~ \sum_{\omega}\ [1 - e^{\eps - \log(P(\omega)/Q(\omega))} ]_{+} \cdot P(\omega)\\
&~=~ \sum_{\eps' \in \supp(\PLD_{(P, Q)})} [1 - e^{\eps - \eps'} ]_{+} \cdot \PLD_{(P, Q)}(\eps')
\end{align*}
where the last line follows from the fact that
\[
\PLD_{(P, Q)}(\eps') ~:=~ \sum_{\omega : \log(P(\omega)/Q(\omega)) = \eps'} P(\omega)\,.
\qedhere
\]
\end{proof}

\section{Evaluation of Poisson-Subsampled Laplace Mechanism}\label{apx:poisson-subsampled-laplace}

Similar to \Cref{sec:exp}, we compute pessimistic and optimistic estimates on the privacy parameter $\eps$, for a fixed value of $\delta$, for varying number of compositions of the Poisson sub-sampled Laplace mechanism and comparing our approach to the Google DP implementation.\footnote{we were unable to compare against Microsoft PRV Accountant, since their implementation does not have support for the Laplace mechanism yet.} In particular, we consider the Laplace mechanism with noise scale $5$, Poisson-subsampled with probability $0.01$. We compare against Google DP implementation twice, once where both algorithms use the same discretization interval, and once where our approach uses a larger discretization interval than the competing algorithm. We additionally plot the running time required for this computation for each number of compositions; we ran the evaluation for each number of compositions $20 \times$ and plot the mean running time along with a shaded region indicating 25th--75th percentiles of running time.

The comparison with Google DP is presented in \Cref{fig:subsamplaplace-pb}.
Figures \ref{fig:subsampgauss-pb-same-disc-epsilon} and \ref{fig:subsamplaplace-pb-same-disc-runtime} compare the $\eps$'s and runtimes for both methods using the same discretization interval, and finds that our method gives a significantly tighter estimate for a mildly larger running time.
Figures \ref{fig:subsamplaplace-pb-diff-disc-epsilon} and \ref{fig:subsamplaplace-pb-diff-disc-runtime} compare the $\eps$'s and runtimes for both methods with different discretization intervals, and using a discretization interval that is $100 \times$ larger, our method gives comparable estimates, with a significant running time speed-up ($\sim 75\times$).

\begin{figure*}
	\centering
	\begin{subfigure}[b]{0.48\textwidth}
		\centering
		\includegraphics[height=0.7\textwidth]{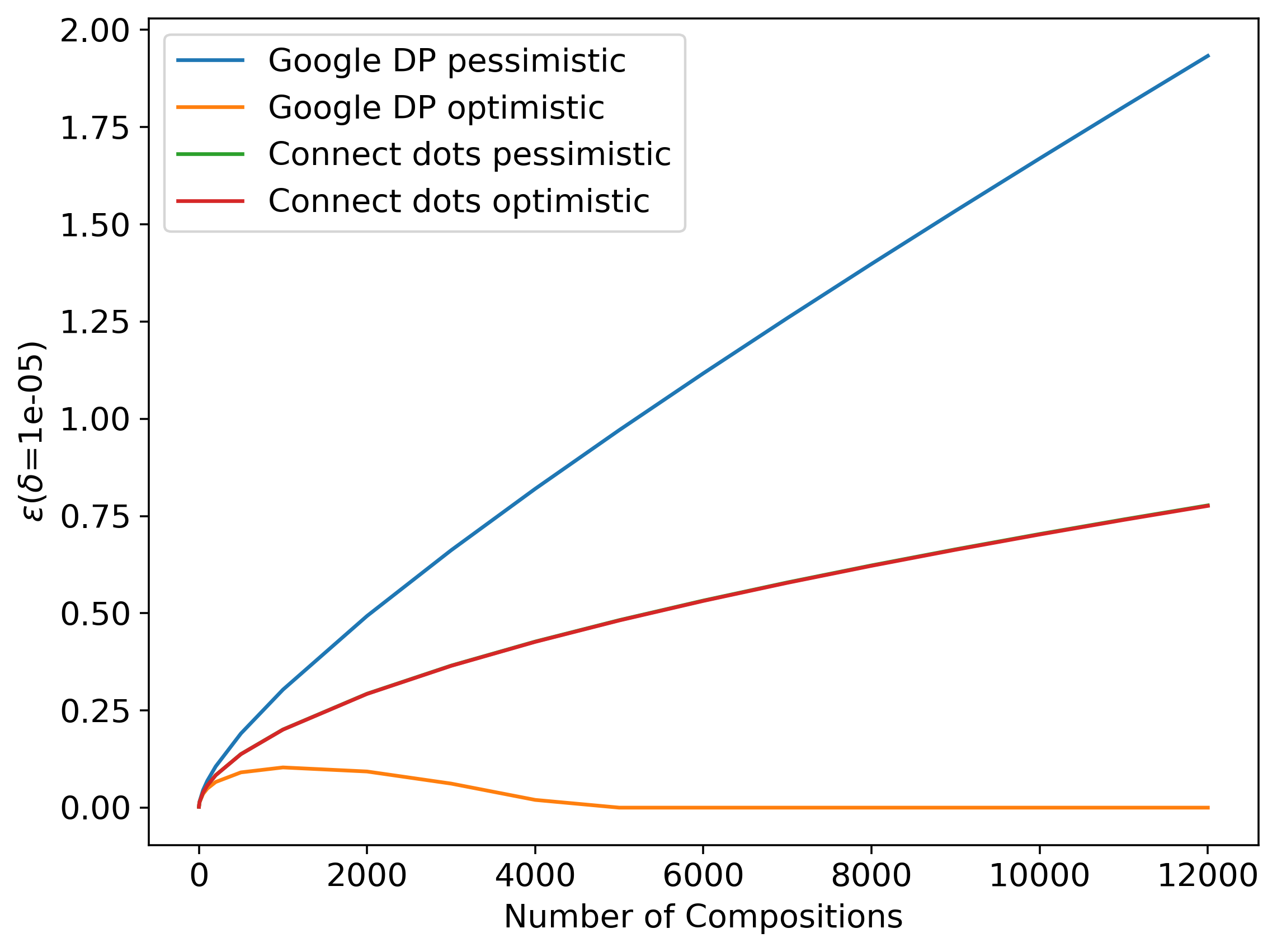}
		\caption{$\eps$ comparison with discretization interval 0.0002 for both this work and Google DP}
		\label{fig:subsamplaplace-pb-same-disc-epsilon}
	\end{subfigure}
	\hfill
	\begin{subfigure}[b]{0.48\textwidth}
		\centering
		\includegraphics[height=0.7\textwidth]{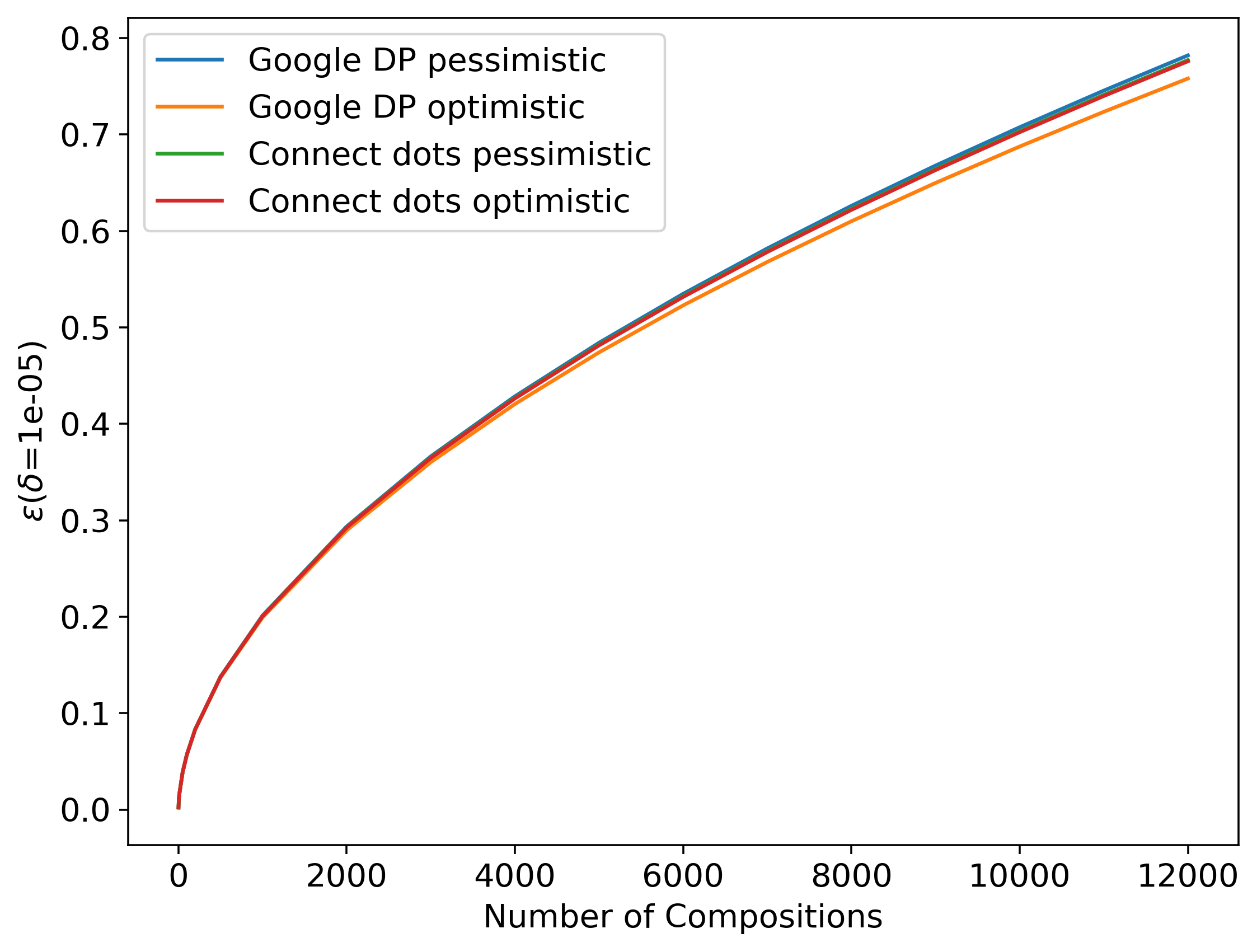}
		\caption{$\eps$ comparison with discretization interval 0.0002 for this work, 0.000002 for Google DP (shaded region indicates 25-75 percentile range across 20 independent runs)}
		\label{fig:subsamplaplace-pb-diff-disc-epsilon}
	\end{subfigure}
	\newline
	\begin{subfigure}[b]{0.48\textwidth}
		\centering
		\includegraphics[height=0.7\textwidth]{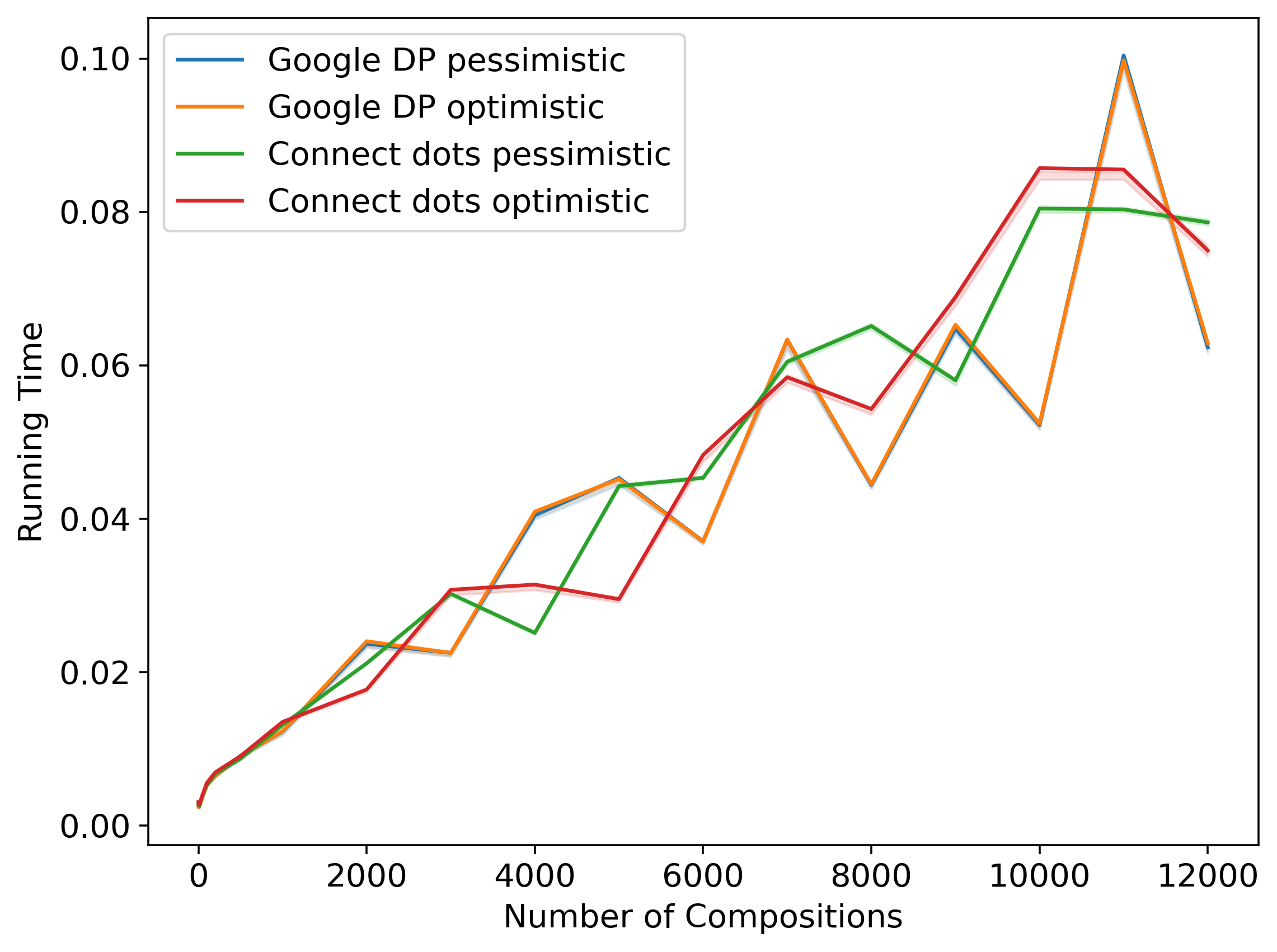}
		\caption{Running time comparison with discretization interval 0.0002 for both this work and Google DP (shaded region indicates 25-75 percentile range across 20 independent runs)}
		\label{fig:subsamplaplace-pb-same-disc-runtime}
	\end{subfigure}
	\hfill
	\begin{subfigure}[b]{0.48\textwidth}
		\centering
		\includegraphics[height=0.7\textwidth]{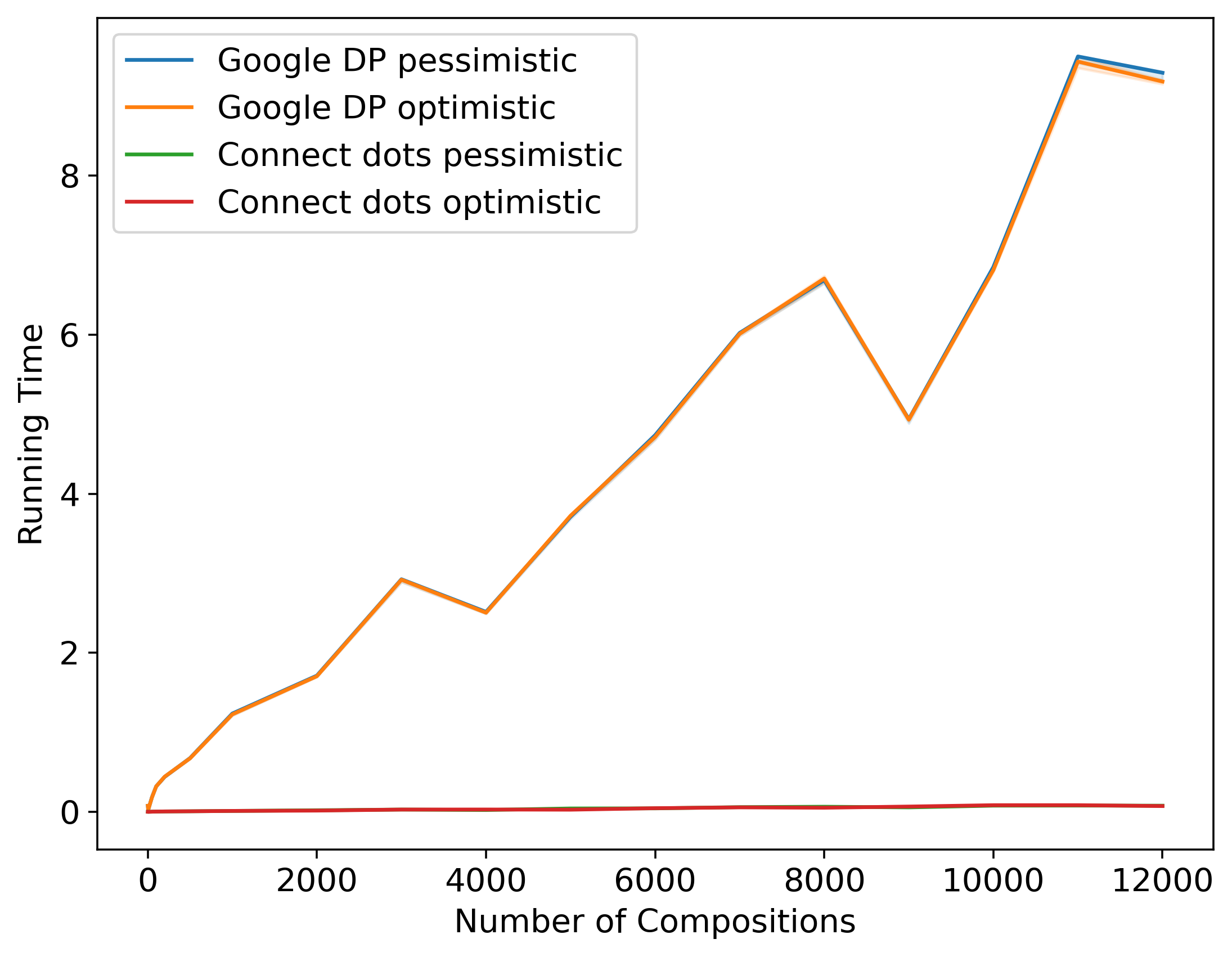}
		\caption{Running time comparison with discretization interval 0.0002 for this work, 0.000002 for Google DP}
		\label{fig:subsamplaplace-pb-diff-disc-runtime}
	\end{subfigure}
	\caption{Computation of pessimistic/optimistic estimates of the privacy parameter $\eps$ (for fixed parameter $\delta = 10^{-5}$) for self-composition of the Laplace mechanism with noise scale $1$, Poisson-subsampled with probability $0.01$, using the Google DP implementation of the PB approach~\cite{meiser2018tight} vs. our approach.%
		Figures \ref{fig:subsamplaplace-pb-same-disc-epsilon} and \ref{fig:subsamplaplace-pb-same-disc-runtime} compare the $\eps$'s and runtimes for both methods using the same discretization interval.
		Figures \ref{fig:subsamplaplace-pb-diff-disc-epsilon} and \ref{fig:subsamplaplace-pb-diff-disc-runtime} compare the $\eps$'s and runtimes for both methods with different discretization intervals.
	}
	\label{fig:subsamplaplace-pb}
\end{figure*}

\section{Inaccuracies from Floating-Point Arithmetic}

We briefly discuss the errors due to floating-point arithmetic. In our implementation, we use the default float datatype in python, which conforms to IEEE-754 ``double precision''. Roughly speaking, this means that the resolution of each floating-point number of $2^{-53} \approx 1.1 \cdot 10^{-16}$. The number of operations performed in our algorithms scales linearly with the support size of the (discretized) PLD, which is less than $10^4$ in all of our experiments. Therefore, a rough heuristic suggests that the numerical error for $\delta$ here would be less than $10^{-11}$. We stress however that this is just a heuristic and is not a formal guarantee: achieving a formal guarantee is much more complicated, e.g., our optimistic algorithm requires computing a convex hull and one would have to formalize how the numerical error from convex hull computation affects the final $\delta$.

Finally, we also remark that, while Gopi et al.~\cite[Appendix A]{gopi2021numerical} note that they experience numerical issues around $\delta \approx 10^{-9}$, we do not experience the same issues in our algorithm even for similar setting of parameters even for $\delta$ as small as $10^{-12}$.

\end{document}

%% file: macros.tex
\usepackage{amsfonts,amssymb,amsmath,amsthm,mathtools}
\usepackage{bm} % For bold symbols
\usepackage{hyperref}
\usepackage{cleveref}
\usepackage{xcolor}
\usepackage{cancel}
\usepackage{caption}
\usepackage{subcaption}
\usepackage{algorithm}
\usepackage[noend]{algpseudocode}
\usepackage{multicol}
\usepackage{thmtools,thm-restate}

\usepackage{pgfplots} % For plotting functions like exp, log, etc.

\usepackage[shortlabels]{enumitem}
\setlist[itemize]{leftmargin=*,label=$\triangleright$}

\usepackage{tikz}
\usetikzlibrary{shapes.geometric,arrows.meta,arrows}

\makeatletter
\def\plist@algorithm{Alg.\space}
\makeatother

\allowdisplaybreaks

\definecolor{Gred}{RGB}{219, 50, 54}
\definecolor{Ggreen}{RGB}{60, 186, 84}
\definecolor{Gblue}{RGB}{72, 133, 237}
\definecolor{Gyellow}{RGB}{247, 178, 16}
\definecolor{ToCgreen}{RGB}{0, 128, 0}
\definecolor{myGold}{RGB}{231,141,20}
\definecolor{myBlue}{rgb}{0.19,0.41,.65}
\definecolor{myPurple}{RGB}{175,0,124}

\providecommand{\Comments}{1}
\ifnum\Comments=1
\usepackage[colorinlistoftodos,prependcaption,textsize=scriptsize]{todonotes}
\paperwidth=\dimexpr \paperwidth + 4.1cm\relax
\oddsidemargin=\dimexpr\oddsidemargin + 2.2cm\relax
\evensidemargin=\dimexpr\evensidemargin + 2.2cm\relax
\fi
\newcommand{\mytodo}[1]{\ifnum\Comments=1{#1}\fi}

\newcommand{\tableoftodos}{\ifnum\Comments=1 \listoftodos[Comments/To Do's] \fi}

\newtheorem{theorem}{Theorem}[section]
\newtheorem{lemma}[theorem]{Lemma}

\newtheorem{corollary}[theorem]{Corollary}

\theoremstyle{definition}
\newtheorem{definition}[theorem]{Definition}
\newtheorem{remark}[theorem]{Remark}
\newtheorem{assumption}[theorem]{Assumption}

\newcommand{\inbrace}[1]{\left \{ #1 \right \}}

\newcommand{\insquare}[1]{\left [ #1 \right ]}

\newcommand{\set}[1]{\inbrace{#1}}

\DeclareMathOperator*{\supp}{supp}

\newcommand{\eps}{\varepsilon}

\newcommand{\bbN}{{\mathbb N}}

\newcommand{\bbR}{{\mathbb R}}

\newcommand{\N}{\bbN}

\newcommand{\R}{\bbR}

\let\boldm\bm
\renewcommand{\bm}{{\boldm m}}

\newcommand{\calA}{\mathcal{A}}
\newcommand{\cA}{\calA}

\newcommand{\calD}{\mathcal{D}}
\newcommand{\calE}{\mathcal{E}}
\newcommand{\cE}{\calE}

\newcommand{\calM}{\mathcal{M}}
\newcommand{\cM}{\calM}

\newcommand{\hP}{\hat{P}}
\newcommand{\hQ}{\hat{Q}}

\newcommand{\of}{\overline{f}}
\newcommand{\oh}{\overline{h}}

%% file: local-macros.tex
\newcommand{\PLD}{\mathsf{PLD}}
\newcommand{\tPLD}{\widetilde{\PLD}}

\tikzstyle{dnode} = [rectangle, rounded corners, minimum width=3cm, minimum height=1cm,text centered, draw=black]
\tikzstyle{arrow} = [thick,->]